\documentclass{lmcs}
\pdfoutput=1

\usepackage[utf8]{inputenc}
\usepackage{lastpage}
\lmcsdoi{17}{3}{3}
\lmcsheading{}{\pageref{LastPage}}{}{}%
{Oct.~30,~2020}{Jul.~20,~2021}{}

\keywords{vector addition system, affine transformation, reachability,
  computational complexity}

\let\C\undefined          %% To avoid errors with complexity package
\usepackage{complexity}   %% Complexity classes
\usepackage{thm-restate}  %% Restatable theorems
\usepackage{macros}       %% Custom macros
\usepackage{amssymb}      %% Math symbols
\usepackage{mathtools}    %% \makemathbox
\usepackage{tikz}         %% TikZ drawings
\usetikzlibrary{automata, backgrounds, fit, positioning}

\begin{document}

\title[The Complexity of Reachability in Affine VASS]{The Complexity
  of Reachability in \\ Affine Vector Addition Systems with States}

\author[M.~Blondin]{Michael Blondin}
\address{Universit\'{e} de Sherbrooke, Canada}
\email{michael.blondin@usherbrooke.ca}
\thanks{M.~Blondin was supported by the Fonds de recherche du Qu\'{e}bec --
  Nature et technologies (FRQNT) and by a Discovery Grant from the
  Natural Sciences and Engineering Research Council of Canada (NSERC)}

\author[M.~Raskin]{Mikhail Raskin}
\address{Technische Universit\"{a}t M\"{u}nchen, Germany}
\email{raskin@in.tum.de}
\thanks{M.~Raskin was upported by the European Research Council (ERC)
  under the European Union's Horizon 2020 research and innovation
  programme under grant agreement No 787367 (PaVeS)}

%%%%%%%%%%%%%%%%%%%%%%%%%%%%%%%%%%%%%%%%%%%%%%%%%%%%%%%%%%%%%%%%%%%%%%%%%%%

\begin{abstract}
  Vector addition systems with states (VASS) are widely used for the
  formal verification of concurrent systems. Given their tremendous
  computational complexity, practical approaches have relied on
  techniques such as reachability relaxations, \eg, allowing for
  negative intermediate counter values. It is natural to question
  their feasibility for VASS enriched with primitives that typically
  translate into undecidability. Spurred by this concern, we pinpoint
  the complexity of integer relaxations with respect to arbitrary
  classes of affine operations.

  More specifically, we provide a trichotomy on the complexity of
  integer reachability in VASS extended with affine operations (affine
  VASS). Namely, we show that it is \NP-complete for VASS with
  resets, \PSPACE-complete for VASS with (pseudo-)transfers and VASS
  with (pseudo-)copies, and undecidable for any other class. We
  further present a dichotomy for standard reachability in affine
  VASS\@: it is decidable for VASS with permutations, and undecidable
  for any other class. This yields a complete and unified complexity
  landscape of reachability in affine VASS\@. We also consider the
  reachability problem parameterized by a fixed affine VASS, rather
  than a class, and we show that the complexity landscape is arbitrary in
  this setting.
\end{abstract}

%% Title header
\maketitle

%% Contents
\section{Introduction}
\label{sec:introduction}
\emph{Vector addition systems with states (VASS)}, which can
equivalently be seen as Petri nets, form a widespread general model of
infinite-state systems with countless applications ranging from the
verification of concurrent programs to the modeling of biological,
chemical and business processes (see, \eg,~\cite{GS92, KKW14, EGLM17,
  HGD08, Aal98}). They comprise a finite-state controller with
counters ranging over $\N$ and updated via instructions of the form $x
\leftarrow x + c$ which are executable if $x + c \geq 0$. The central
decision problem concerning VASS is the \emph{reachability problem}:
given configurations $x$ and $y$, is it possible to reach $y$ starting
from $x$? Such queries allow, \eg, to verify whether unsafe states can
be reached in concurrent programs. The notorious difficulty of the
reachability problem led to many proofs of its decidability over the
last decades~\cite{ST77, May81, Kos82, Lam92, Ler10, Ler11,
  Ler12}. While the problem has been known to be \EXPSPACE-hard since
1976~\cite{Lip76}, its computational complexity has remained unknown
until very recently, where it was shown to be
\TOWER-hard~\cite{CLLLM19} and solvable in Ackermannian
time~\cite{LS15, LS19}.

Given the potential applications on the one hand, and the tremendously
high complexity on the other hand, researchers have investigated
relaxations of VASS in search of a tradeoff between expressiveness and
algorithmic complexity. Two such relaxations consist in permitting
either:
\begin{enumerate}[(a)]
\item transitions to be executed fractionally, and consequently
  counters to range over $\Q_{\geq 0}$ (\emph{continuous
    reachability}); or

\item counters to range over $\Z$ (\emph{integer reachability}).
\end{enumerate}
In both cases, the complexity drops drastically: continuous
reachability is \P-complete and \NP-complete for Petri nets and VASS
respectively~\cite{FH15, BH17}, while integer reachability is
\NP-complete for both models~\cite{HH14,CHH18}. Moreover, these two types of
reachability have been used successfully to prove safety of real-world
instances like multithreaded program skeletons,
\eg, see~\cite{ELMMN14, ALW16, BFHH17}.

Although VASS are versatile, they are sometimes too limited to model
common primitives. Consequently, their modeling power has been
extended with various operations. For example,
\emph{(multi-)transfers}, \ie, operations of the form \[x \leftarrow x
+ \sum_{i=1}^n y_i;\ y_1 \leftarrow 0;\ y_2 \leftarrow
0;\ \cdots;\ y_n \leftarrow 0,\] allow, \eg, for the verification of
multi-threaded C and Java program skeletons with communication
primitives~\cite{KKW14, DRV02}. Another example is the case of
\emph{resets}, \ie, operations of the form $x \leftarrow 0$, which
allow, \eg, for the validation of some business
processes~\cite{WAHE09}, and the generation of program loop
invariants~\cite{SK19}. Many such extensions fall under the generic
family of \emph{affine VASS}, \ie, VASS with instructions of the form
$\vec{x} \leftarrow \mat{A} \cdot \vec{x} + \vec{b}$. As a general
rule of thumb, reachability is undecidable for essentially any class
of affine VASS introduced in the literature; in particular, for
transfers and resets~\cite{AK76, DFS98}.

Given the success of relaxations for the practical analysis of
(standard) VASS, it is tempting to employ the same approach for affine
VASS\@. Unfortunately, continuous reachability becomes undecidable for
mild affine extensions such as resets and transfers. However,
\emph{integer reachability} was recently shown decidable for affine
operations such as resets (\NP-complete) and transfers
(\PSPACE-complete)~\cite{HH14, BHM18}. While such complexity results
do not translate immediately into practical procedures, they arguably
guide the design of algorithmic verification strategies.

\medskip\parag{Contribution} Thus, these recent results raise two
natural questions: for \emph{what classes} of affine VASS is integer
reachability decidable? And, whenever it is decidable, what is its
\emph{exact} computational complexity? We fully answer these questions
in this paper by giving a precise trichotomy: integer reachability is
\NP-complete for VASS with resets, \PSPACE-complete for VASS with
(pseudo-)transfers and VASS with (pseudo-)copies, and undecidable for
\emph{any} other class. In particular, this answers a question left
open in~\cite{BHM18}: integer reachability is undecidable for
\emph{any} class of affine VASS with \emph{infinite} matrix monoid.

This clear complexity landscape is obtained by formalizing classes of
affine VASS and by carefully analyzing the structure of arbitrary
affine transformations; which could be of independent interest. In
particular, it enables us to prove a dichotomy on \emph{(standard)
reachability} for affine VASS\@: it is decidable for VASS with
permutations, and undecidable for \emph{any} other class. To the best
of our knowledge, this is the first proof of the folkore rule of thumb
stating that ``reachability is undecidable for essentially any class
of affine VASS''.

We further complement these trichotomy and dichotomy by showing that
the (integer or standard) reachability problem has an arbitrary
complexity when it is parameterized by a fixed affine VASS rather than
a fixed class.

\medskip\parag{Related work} Our work is related
to~\cite{BHM18} which shows that integer reachability is decidable for
affine VASS whose matrix monoid is finite (refined to \EXPSPACE\
by~\cite{BHKST20}); and more particularly \PSPACE-complete in general
for VASS with transfers and VASS with copies. While it is also
recalled in~\cite{BHM18} that integer reachability is undecidable in
general for affine VASS, the authors do not provide any necessary
condition for undecidability to hold. Moreover, the complexity
landscape for affine VASS with finite monoids is left blurred, \eg, it
does not give necessary conditions for \PSPACE-hardness results to
hold, and the complexity remains unknown for monoids with negative
coefficients. This paper completes the work initiated in~\cite{BHM18}
by providing a unified framework, which includes the notion
of \emph{matrix class}, that allows us to precisely characterize the
complexity of integer reachability for \emph{any} class of affine
VASS\@.

Our work is also loosely related to a broader line of research on
(variants of) affine~VASS dealing with, \eg, modeling
power~\cite{Valk78}, accelerability~\cite{FL02}, formal
languages~\cite{CFM12}, coverability~\cite{BFP12}, and the complexity
of integer reachability for restricted counters~\cite{FGH13} and
structures~\cite{IS16}.

\medskip\parag{Structure of the paper} Section~\ref{sec:preliminaries}
introduces general notation and affine VASS\@. In
Section~\ref{sec:trichotomy}, we prove our main result, namely the
complexity trichotomy for integer reachability in affine VASS\@. In
Section~\ref{sec:dichotomy}, we show a dichotomy on (standard)
reachability in affine VASS\@. In Section~\ref{sec:arb:complexity}, we
show that the complexity of reachability can be arbitrary when fixing
an affine VASS\@. Finally, we conclude in
Section~\ref{sec:conclusion}. To avoid cluttering the presentation
with too many technical details, some of them are deffered to an
appendix.

\section{Preliminaries}
\label{sec:preliminaries}
\parag{Notation} Let $\Z$, $\N$, $[a, b]$ and $[k]$ denote
respectively the sets $\{\ldots, -1, 0, 1, \ldots\}$, $\{0, 1, 2,
\ldots\}$, $\{a, a + 1, \ldots, b\}$ and $[1, k]$. For every vectors
$\vec{u}, \vec{v} \in \Z^k$, let $\vec{u} + \vec{v}$ be the vector
$\vec{w} \in \Z^k$ such that $\vec{w}(i) \defeq \vec{u}(i) +
\vec{v}(i)$ for every $i \in [k]$. Let $\vec{e}_i$ be the unit vector
such that $\vec{e}_i(i) = 1$ and $\vec{e}_i(j) = 0$ for every $j \neq
i$. We do not specify the arity of $\vec{e}_i$ as we will use it
without ambiguity in various dimensions. For every square matrix
$\mat{A} \in \Z^{k \times k}$, let $\size{\mat{A}} \defeq k$ and
let \[\norm{\mat{A}} \defeq \max\{|\mat{A}_{i, j}| : i, j \in [k]\}.\]
We naturally extend the latter notation to any set $X$ of matrices,
\ie, $\norm{X} \defeq \sup\{\norm{\mat{A}} : \mat{A} \in
X\}$. Throughout the paper, we will often refer to matrix and vector
indices as \emph{counters}. We will also often describe permutations
in cycle notation, where elements are separated by semicolons for
readability, \eg, $(i; j)$ denotes the permutation that swaps $i$ and
$j$.

\medskip\parag{Affine VASS} An \emph{affine vector
  addition system with states (affine VASS)} is a tuple $\V = (d, Q,
T)$ where:
\begin{itemize}
\item $d \geq 1$ is the \emph{number of counters} of $\V$;

\item $Q$ is a finite set of elements called
  \emph{control-states};

\item $T \subseteq Q \times \Z^{d \times d} \times \Z^d \times Q$ is a
  finite set of elements called \emph{transitions}.
\end{itemize}

For every transition $t = (p, \mat{A}, \vec{b}, q)$, let $\tsrc{t}
\defeq p$, $\tmat{t} \defeq \mat{A}$, $\tvec{t} \defeq \vec{b}$ and
$\ttgt{t} \defeq q$. A \emph{configuration} is a pair $(q, \vec{v})
\in Q \times \Z^d$ written $q(\vec{v})$. For all $t \in T$ and
$\D \in \{\Z, \N\}$, we~write \[p(\vec{u}) \trans{t}_\D q(\vec{v})\]
if $\vec{u}, \vec{v} \in \D^d$, $\tsrc{t} = p$, $\ttgt{t} = q$, and
$\vec{v} = \tmat{t} \cdot \vec{u} + \tvec{t}$. The relation
$\trans{}_\D$ is naturally extended to sequences of transitions,
\ie, for every $w \in T^k$ we
let \[\trans{w}_\D\ \defeq\ \trans{w_k}_\D \circ \cdots \circ
\trans{w_2}_\D \circ \trans{w_1}_\D.\] Moreover, we write
\begin{align*}
  p(\vec{u}) \trans{}_\D q(\vec{v})\ &
  \text{ if $p(\vec{u}) \trans{\mathmakebox[8pt][c]{t}}_\D q(\vec{v})$
  for some $t \in T$, and} \\  
  p(\vec{u}) \trans{*}_\D q(\vec{v})\ &
  \text{ if $p(\vec{u}) \trans{\mathmakebox[8pt][c]{w}}_\D q(\vec{v})$
  for some $w \in T^*$}.
\end{align*}

As an example, let us consider the affine VASS of
Figure~\ref{fig:ex:vass},
\ie, where $d = 2$, $Q = \{p, q, r\}$ and $T$ is as depicted
graphically. We have:
\[
  p(3, 1) \trans{s}_\Z q(1, 0)
  \trans{t}_\Z r(1, 0) \trans{u}_\Z q(1, 1)
  \trans{t}_\Z r(2, 0) \trans{u}_\Z q(2, 2) \trans{t}_\Z r(4, 0).
\]
More generally, $p(x, 1) \trans{*}_\Z r(2^k, 0)$ for all $x \in \Z$
and $k \in \N_{> 0}$. However, $p(3, 0) \trans{s}_\N q(1, -1)$
does \emph{not} hold as counters are not allowed to become negative
under this semantics.

\begin{figure}[h]
  \centering
  \begin{tikzpicture}[->, node distance=5cm, auto, very thick, scale=0.8, transform shape, font=\Large]      
    %% Control-states
    \tikzset{every state/.style={inner sep=1pt, minimum size=25pt}}

    \node[state] (p)               {$p$};
    \node[state] (q) [right= of p] {$q$};
    \node[state] (r) [right= of q] {$r$};

    %% Transitions
    \path[->]
    (p) edge[] node[] {
      $
      s \colon
      \begin{pmatrix}
        0 & 0 \\
        0 & 1
      \end{pmatrix},
      \begin{pmatrix}
       1 \\
       -1
      \end{pmatrix}
      $
    } (q)
    
    (q) edge[bend left=20] node[] {
      $
      t \colon
      \begin{pmatrix}
        1 & 1 \\
        0 & 0
      \end{pmatrix},
      \begin{pmatrix}
        0 \\
        0
      \end{pmatrix}
      $
    } (r)
    
    (r) edge[bend left=20] node[] {
      $
      u \colon
      \begin{pmatrix}
        1 & 0 \\
        1 & 0
      \end{pmatrix},
      \begin{pmatrix}
        0 \\
        0
      \end{pmatrix}
      $
    } (q)
    ;
  \end{tikzpicture}
  \caption{Example of an affine VASS.}\label{fig:ex:vass}
\end{figure}
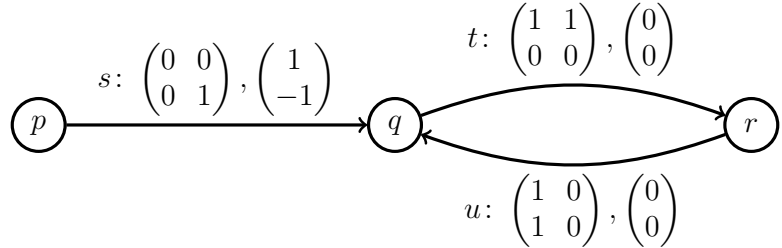

\medskip\parag{Classes of matrices} Let us formalize the informal
notion of \emph{classes} of affine VASS, such as ``VASS with resets'',
``VASS with transfers'', ``VASS with doubling'', etc., used throughout
the literature.

Such classes depend on the extra operations they provide, \ie, by
their affine transformations, and more precisely by their linear part
(matrices) rather than their additive part (vectors). Affine VASS extend
standard VASS since they always include the identity matrix, which
amounts to not applying any extra operation. Moreover, as
transformations can be composed along sequences of transitions, their
matrices are closed under multiplication, \ie, if matrices $\tmat{s}$
and $\tmat{t}$ are allowed on transitions $s$ and $t$ of two affine
VASS of a given class, then $\tmat{s} \cdot \tmat{t}$ is typically
also allowed in an affine VASS of the class as it is understood that
$s$ and $t$ can be composed. In other words, matrices form
a \emph{monoid}. In addition, classes of affine VASS typically
considered do not pose restrictions on the number of counters that can
be used, or on the subset of counters on which operations can be
applied. In other words, their affine transformations can be extended
to \emph{arbitrary dimensions} and can be applied on \emph{any subset
of counters}, \eg, general ``VASS with resets'' allow to reset any
counter, not just say the first one.

We formalize these observations as follows. For every $k \geq 1$, let
$\mat{I}_k$ be the $k \times k$ identity matrix and let $\S_k$ denote
the set of permutations over $[k]$. Let $\mat{P}_\sigma \in {\{0,
1\}}^{k \times k}$ be the permutation matrix of $\sigma \in \S_k$. For
every matrix $\mat{A} \in \Z^{k \times k}$, every permutation $\sigma
\in \S_k$ and every $n \geq 1$, let \[\perm{\mat{A}}{\sigma} \defeq
\mat{P}_\sigma \cdot \mat{A} \cdot \mat{P}_{\sigma^{-1}},\] and let
$\ext{\mat{A}}{n} \in \Z^{(k + n) \times (k + n)}$ be the matrix such
that:
\[
\ext{\mat{A}}{n} \defeq
\begin{pmatrix}
  \mat{A} & \mat{0}_{\phantom{n}} \\
  \mat{0} & \mat{I}_n
\end{pmatrix}.
\]

A \emph{class (of matrices)} is a set of matrices $\C \subseteq
\bigcup_{k \geq 1} \Z^{k \times k}$ that satisfies
$\{\perm{\mat{A}}{\sigma}, \ext{\mat{A}}{n}, \allowbreak \mat{I}_n,
\mat{A} \cdot \mat{B}\} \subseteq \C$ for every $\mat{A}, \mat{B} \in
\C$, every $\sigma \in \S_{\dim{\mat{A}}}$ and every $n \geq 1$. In
other words, $\C$ is closed under counter renaming; each matrix of
$\C$ can be extended to larger dimensions; and $\C \cap \Z^{k \times
  k}$ is a monoid under matrix multiplication for every $k \geq 1$.

Note that ``counter renaming'' amounts to choosing a set of counters
on which to apply a given transformation, \ie, it renames the
counters, applies the transformation, and renames the counters back to
their original names. Let us illustrate this. Consider the classical
case of transfer VASS, \ie, where the contents of a counter can be
transferred onto another counter with operations of the form ``$x
\leftarrow x + y;\ y \leftarrow 0$''. In matrix notation, this
amounts~to:
\[
\mat{O} \defeq
\begin{pmatrix}
1 & 1 \\
0 & 0 
\end{pmatrix}.
\] Now, consider a system with three counters $c_1$, $c_2$ and
$c_3$. This system should be able to compute ``$c_1 \leftarrow c_1 +
c_2 + c_3;\ c_2 \leftarrow 0;\ c_3 \leftarrow 0$'', but matrix
$\mat{O}$ cannot achieve this on its own. However, it can be done with
the following matrix:
\[
\mat{O}' \defeq
\begin{pmatrix}
  1 & 1 & \textit{0} \\
  0 & 0 & \textit{0} \\
  \textit{0} & \textit{0} & \textit{1}
\end{pmatrix}
\cdot
\begin{pmatrix}
  1 & \textit{0} & 1 \\
  \textit{0} & \textit{1} & \textit{0} \\
  0 & \textit{0} & 0
\end{pmatrix}
=
\begin{pmatrix}
  1 & 1 & 1 \\
  0 & 0 & 0 \\
  0 & 0 & 0
\end{pmatrix}.
\]
We have $\mat{O}' = \ext{\mat{O}}{1} \cdot
\perm{\ext{\mat{O}}{1}}{\sigma}$ where $\sigma \defeq (2; 3)$. Thus,
the operation can be achieved by any class containing $\mat{O}$. The
symmetric operation ``$c_3 \leftarrow c_1 + c_2 + c_3;\ c_1 \leftarrow
0;\ c_2 \leftarrow 0$'', \eg, can also be achieved with appropriate
permutations. Hence, this corresponds to the usual notion of
transfers: we are allowed to choose some counters and apply transfers
in either direction.

Note that requiring $\mat{P}_\sigma \cdot \mat{A} \in \C$ for classes
would be too strong as it would allow to permute the contents of
counters even for classes with no permutation matrix, such as resets.

\medskip\parag{Classes of interest} We say that a matrix $\mat{A} \in
\Z^{k \times k}$ is a \emph{pseudo-reset}, \emph{pseudo-transfer} or
\emph{pseudo-copy} matrix if $\mat{A} \in {\{-1, 0, 1\}}^{k \times k}$
and if it also satisfies the following:
\begin{itemize}
\item \makebox[4cm][l]{pseudo-reset matrix:} $\mat{A}$ is a diagonal matrix;

\item \makebox[4cm][l]{pseudo-transfer matrix:} $\mat{A}$ has at most one nonzero entry per column;

\item \makebox[4cm][l]{pseudo-copy matrix:} $\mat{A}$ has at most one nonzero entry per row.
\end{itemize}

We omit the prefix ``pseudo-'' if $\mat{A} \in {\{0, 1\}}^{k \times
  k}$. Note that the sets of (pseudo-)reset matrices,
  (pseudo-)transfer matrices, and (pseudo-)copy matrices all form
  classes. Moreover, (pseudo-)reset matrices are both
  (pseudo-)transfer and (pseudo-)copy matrices.

Note that the terminology of ``reset'', ``transfer'' and ``copy''
comes from the fact that such matrices implement operations like ``$x
\leftarrow 0$'', ``$x \leftarrow x + y;\ y \leftarrow 0$'' and ``$x
\leftarrow x;\ y \leftarrow x$'', as achieved respectively by the
matrices of transitions $s$ (\emph{reset}), $t$ (\emph{transfer}) and
$u$ (\emph{copy}) illustrated in Figure~\ref{fig:ex:vass}.

\medskip\parag{Reachability problems} We say that an affine VASS $\V =
(k, Q, T)$ \emph{belongs to} a class of matrices $\C$ if $\{\tmat{t} :
t \in T\} \subseteq \C$, \ie, if all matrices appearing on its
transitions belong to $\C$. The \emph{reachability problem} and
\emph{integer reachability problem} for a fixed class $\C$ are defined
as:
\begin{center}
  \begin{tabular}{lp{12cm}}
    \multicolumn{2}{l}{\underline{$\reach{\C}$}} \\[5pt]
  
    \textsc{Input}: & an affine VASS $\V$ that belongs to $\C$, and
    two configurations $p(\vec{u}), q(\vec{v})$; \\

    \textsc{Decide}: & $p(\vec{u}) \trans{*}_\N q(\vec{v})$ in $\V$?
  \end{tabular}

  \bigskip
  
  \begin{tabular}{lp{12cm}}
    \multicolumn{2}{l}{\underline{$\zreach{\C}$}} \\[5pt]

    \textsc{Input}: & an affine VASS $\V$ that belongs to $\C$, and
    two configurations $p(\vec{u}), q(\vec{v})$; \\

    \textsc{Decide}: & $p(\vec{u}) \trans{*}_\Z q(\vec{v})$ in $\V$?
  \end{tabular}
\end{center}

\section{A complexity trichotomy for integer reachability}
\label{sec:trichotomy}
This section is devoted to the proof of our main result, namely the
trichotomy on $\zreach{\C}$:

\begin{thm}\label{thm:trichotomy}
  The integer reachability problem $\zreach{\C}$~is:
  \begin{enumerate}[(i)]
  \item \NP-complete if $\C$ only contains reset matrices;\label{itm:np}

  \item \PSPACE-complete, otherwise, if either $\C$ only contains
    pseudo-transfer matrices or $\C$ only contains pseudo-copy
    matrices;\label{itm:pspace}
    
  \item Undecidable otherwise.\label{itm:undec}
  \end{enumerate}
\end{thm}

It is known from~\cite[Cor.~10]{HH14} that \NP-hardness holds for
affine VASS using only the identity matrix (recall that our definition
of reset matrices include the identity matrix), and that
\NP\ membership holds for any class of reset matrices. Hence,
\ref{itm:np} follows immediately. Thus, the rest of this section is
dedicated to proving~\ref{itm:pspace} and~\ref{itm:undec}.

\subsection{PSPACE-hardness}

For the rest of this subsection, let us fix some class $\C$ that
either only contains pseudo-transfer matrices or only contains
pseudo-copy matrices. We prove \PSPACE-hardness of $\zreach{\C}$ by
first proving that \PSPACE-hardness holds if either:
\begin{itemize}
\item $\C$ contains a matrix with an entry equal to $-1$; or

\item $\C$ contains a matrix with entries from $\{0, 1\}$ and a
  nonzero entry outside of its diagonal.
\end{itemize}

For these two cases, we first show that $\C$ can implement operations
$x \leftarrow -x$ or $(x, y) \leftarrow (y, x)$ respectively,
\ie, \emph{sign flips} or \emph{swaps}. Essentially, each of these
operations is sufficient to simulate linear bounded automata. Before
investigating these two cases, let us carefully formalize what it
means to implement an operation:

\begin{defi}\label{def:op:impl}
  Let $f \colon \Z^k \mapsto \Z^k$ and let $\tau \in \{0, ?\}$. Given
  a set of counters $X \subseteq [m]$,
  let \[V_X \defeq \left\{\vec{v} \in \Z^m : \bigwedge_{j \not\in
  X} \vec{v}(j) = 0\right\} \text{ if } \tau = 0,\] and let
  $V_X \defeq \Z^m$ otherwise. We say that
  $\C$ \emph{$\tau$-implements} $f$ if for every $n \geq k$, there
  exist counters $X = \{x_1, x_2, \ldots, x_n\}$, matrices
  $\{\mat{F}_\sigma : \sigma \in \S_k\} \subseteq \C$ and $m \geq n$
  such that the following holds for every $\sigma \in \S_k$ and
  $\vec{v} \in V_X$:
  \begin{enumerate}[(a)]
  \item $\size{\mat{F}_\sigma} = m$;\label{itm:op:dim}
 
  \item $(\mat{F}_\sigma \cdot \vec{v})(x_{\sigma(i)}) =
    f(x_{\sigma(1)}, x_{\sigma(2)}, \ldots, x_{\sigma(k)})(i)$ for
    all $i \in [k]$;\label{itm:op:f}

  \item $(\mat{F}_\sigma \cdot \vec{v})(x_{\sigma(i)}) =
    \vec{v}(x_{\sigma(i)})$ for every $i \not\in
        [k]$;\label{itm:op:untouched}

  \item $\mat{F}_\sigma \cdot \vec{v} \in V_X$.\label{itm:op:closure}
  \end{enumerate}
  We further say that $\C$ \emph{implements} $f$ if it either
  $0$-implements or $?$-implements $f$.
\end{defi}

Definition~\ref{def:op:impl}~\ref{itm:op:f} and~\ref{itm:op:untouched}
state that it is possible to obtain arbitrarily many counters $X$ such
that $f$ can be applied on any $k$-subset of $X$, provided that the
counter values belong to $V_X$. Moreover, \ref{itm:op:closure} states
that vectors resulting from applying operation $f$ also belong to
$V_X$, which ensures that $f$ can be applied arbitrarily many
times. Note that~\ref{itm:op:dim} allows for extra auxiliary counters
whose values are only restricted by $V_X$.

Informally, $?$-implementation means that we use additional counters
that can hold arbitrary values, while $0$-imple\-mentation requires
the extra counters to be initialized with zeros but promises to keep
them in this state. It turns out that pseudo-transfer matrix classes
$0$-implement the functions we need, while pseudo-copy matrix classes
$?$-implement~them.

\begin{prop}\label{prop:exist:flip}
  If $\C$ contains a matrix with some entry equal to $-1$, then it
  implements sign flips, \ie, the operation $f \colon \Z \to \Z$ such
  that $f(x) \defeq -x$.
\end{prop}
  
\begin{proof}
  Let $n \geq 1$ and $\mat{A} \in \C$ be such that $\mat{A}_{a, b} =
  -1$ for some counters $a$ and $b$. Let $d \defeq \size{\mat{A}}$. We
  extend $\mat{A}$ with $n + 2$ counters $X' \defeq X \cup \{y, z\}$,
  where $X \defeq \{x_i : i \in [n]\}$ are the counters for which we
  wish to implement sign flips, and $\{y, z\}$ are auxiliary
  counters. More formally, let $\mat{A}' \defeq \ext{\mat{A}}{n + 2}$
  where $X' = [d + 1, d']$ and $d' \defeq d + n + 2$.

  For every $s, t \in X'$ such that $s \neq t$, let $\mat{B}_{s, t}
  \defeq \pi_{s, t}(\mat{A}')$ and let $\mat{C}_t \defeq
  \sigma_t(\mat{A}')$ where $\pi_{s, t} \defeq (a; t)(b; s)$ and
  $\sigma_t \defeq (a; t)$. For every $x \in X$, let
  \[
  \mat{F}_x \defeq
  \begin{cases}
    \mat{B}_{z, x} \cdot \mat{B}_{y, z} \cdot \mat{B}_{x, y} &
    \text{if $a \neq b$}, \\
    \mat{C}_x & \text{otherwise}.
  \end{cases}
  \]

  Intuitively, $\mat{B}_{s, t}$ (resp.\ $\mat{C}_t$) flips the sign
  from source counter $s$ (resp.\ $t$) to target counter $t$. If $a
  \neq b$, then matrix $\mat{F}_x$ implements a sign flip in three
  steps using auxiliary counters $y$ and $z$, as illustrated in
  Figure~\ref{fig:flip}. Otherwise, $\mat{F}_x$ implements sign flip
  directly in one step.

  \begin{figure*}[!h]
  \hfill
  %% Pseudo-transfer
  \begin{tikzpicture}[->, node distance=0.75cm, auto, very thick, scale=0.8, transform shape, font=\Large]
    %% Counters
    \tikzset{every state/.style={inner sep=1pt, minimum size=5pt}}
    \newcommand{\cspace}{2.25cm}

    \node[state, label={right:$x$}]       (xi0)                {};
    \node[state, label={right:$y$}, fill] (yi0) [below=of xi0] {};
    \node[state, label={right:$z$}, fill] (zi0) [below=of yi0] {};

    \node[state, fill] (xi1) [left=\cspace of xi0] {};
    \node[state]       (yi1) [below=       of xi1] {};
    \node[state, fill] (zi1) [below=       of yi1] {};

    \node[state, fill] (xi2) [left=\cspace of xi1] {};
    \node[state, fill] (yi2) [below=       of xi2] {};
    \node[state]       (zi2) [below=       of yi2] {};

    \node[state]       (xi3) [left=\cspace of xi2] {};
    \node[state, fill] (yi3) [below=       of xi3] {};
    \node[state, fill] (zi3) [below=       of yi3] {};

    %% Opérations
    \path[->, dashed]
    (xi0) edge node {} (yi1)
    (yi1) edge node {} (zi2)
    (zi2) edge node {} (xi3)
    ;

    \path[->]
    (yi0) edge node {} (xi1)
    (zi0) edge node {} (zi1)

    (xi1) edge node {} (xi2)
    (zi1) edge node {} (yi2)

    (yi2) edge node {} (yi3)
    (xi2) edge node {} (zi3)
    ;

    \node[right={\dimexpr\cspace/8} of zi1, yshift=-20pt]
         {$\mat{B}_{x, y}$};
    \node[right={\dimexpr\cspace/8} of zi2, yshift=-20pt]
         {$\mat{B}_{y, z}$};
    \node[right={\dimexpr\cspace/8} of zi3, yshift=-20pt]
         {$\mat{B}_{z, x}$};
  \end{tikzpicture}%
  \hfill\hfill
  %% Pseudo-copy
  \begin{tikzpicture}[->, node distance=0.75cm, auto, very thick, scale=0.8, transform shape, font=\Large]
    %% Counters
    \tikzset{every state/.style={inner sep=1pt, minimum size=5pt}}
    \newcommand{\cspace}{2.25cm}

    \node[state, label={right:$x$}] (xi0)                {};
    \node[state, label={right:$y$}] (yi0) [below=of xi0] {};
    \node[state, label={right:$z$}] (zi0) [below=of yi0] {};

    \node[state] (xi1) [left=\cspace of xi0] {};
    \node[state] (yi1) [below=       of xi1] {};
    \node[state] (zi1) [below=       of yi1] {};

    \node[state] (xi2) [left=\cspace of xi1] {};
    \node[state] (yi2) [below=       of xi2] {};
    \node[state] (zi2) [below=       of yi2] {};

    \node[state] (xi3) [left=\cspace of xi2] {};
    \node[state] (yi3) [below=       of xi3] {};
    \node[state] (zi3) [below=       of yi3] {};

    \node (zr0) [left={\dimexpr\cspace/2} of xi0] {?};
    \node (zr1) [left={\dimexpr\cspace/2} of zi0] {?};
    \node (zr2) [left={\dimexpr\cspace/2} of xi1] {?};
    \node (zr3) [left={\dimexpr\cspace/2} of yi1] {?};
    \node (zr4) [left={\dimexpr\cspace/2} of yi2] {?};
    \node (zr5) [left={\dimexpr\cspace/2} of zi2] {?};

    %% Opérations
    \path[->, dashed]
    (xi0) edge node {} (yi1)
    (yi1) edge node {} (zi2)
    (zi2) edge node {} (xi3)
    ;

    \path[->]
    (zr0) edge node {} (xi1)
    (zr1) edge node {} (zi1)

    (zr2) edge node {} (xi2)
    (zr3) edge node {} (yi2)

    (zr4) edge node {} (yi3)
    (zr5) edge node {} (zi3)
    ;

    \node[right={\dimexpr\cspace/8} of zi1, yshift=-20pt]
         {$\mat{B}_{x, y}$};
    \node[right={\dimexpr\cspace/8} of zi2, yshift=-20pt]
         {$\mat{B}_{y, z}$};
    \node[right={\dimexpr\cspace/8} of zi3, yshift=-20pt]
         {$\mat{B}_{z, x}$};
  \end{tikzpicture}%
  \hfill\hfill
  \caption{Effect of applying $\mat{F}_x$ for the case $a \neq b$,
    where the left (resp.\ right) diagram depicts the case where
    $\mat{A}$ is a pseudo-transfer (resp.\ pseudo-copy) matrix. A
    solid or dashed edge from $s$ to $t$ represents
    operation $s \leftarrow t$ or $s \leftarrow -t$
    respectively. Filled nodes indicate counters that necessarily hold
    $0$. Symbol ``?''  stands for an integer whose value is irrelevant
    and depends on $\mat{A}$ and the counter values.}\label{fig:flip}
\end{figure*}
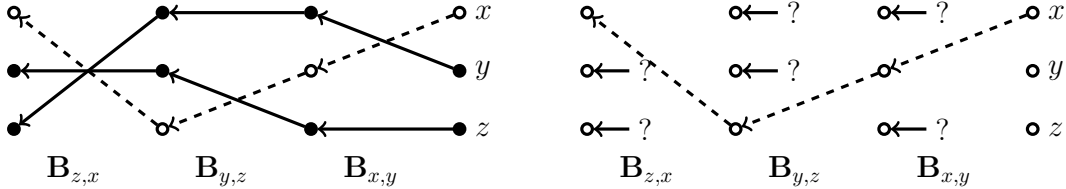

  Let us consider the case where $\mat{A}$ is a pseudo-transfer
  matrix. From the definition of $\mat{B}_{s, t}$ and $\mat{C}_t$, it
  can be shown that for every $s, t, u \in X'$ such that $s \neq t$ and
  $u \not\in \{s, t\}$, the following holds:
  \begin{enumerate}[(i)]
  \item $\mat{B}_{s, t} \cdot \vec{e}_s = \mat{C}_t \cdot \vec{e}_t =
    -\vec{e}_t$, and\label{itm:flips:src:tgt}
    
  \item $\mat{B}_{s, t} \cdot \vec{e}_u = \mat{C}_t \cdot \vec{e}_u =
    \vec{e}_u$.\label{itm:flips:other}
  \end{enumerate}

  Let us show that we $0$-implement sign flips, so let \[V \defeq
  \left\{\vec{v} \in \Z^{d'} : \bigwedge_{j \not\in X}
  \vec{v}(j)=0\right\}.\] Let $\vec{v} \in V$ and $x \in X$. By
  definition of $V$, $\vec{v} = \sum_{y \in X} \vec{v}(y) \cdot
  \vec{e}_y$. Let $\vec{v}' \defeq \sum_{j \in X \setminus \{x\}}
  \vec{v}(j) \cdot \vec{e}_j$. Items~\ref{itm:op:f},
  \ref{itm:op:untouched} and~\ref{itm:op:closure} of
  Definition~\ref{def:op:impl} are satisfied since:
  \begin{align*}
    \mat{F}_x \cdot \vec{v}
    &= \mat{F}_x \cdot \vec{v}(x) \cdot \vec{e}_x + \mat{F}_x \cdot
    \vec{v}' \\    
    &= \vec{v}(x) \cdot \mat{F}_x \cdot \vec{e}_x + \vec{v}'
    && \text{(by~\ref{itm:other} and def.\ of $\mat{F}_x$)} \\
    &= \vec{v}(x) \cdot -\vec{e}_x + \vec{v}'
    && \text{(by~\ref{itm:flips:src:tgt}, \ref{itm:other} and
      def.\ of $\mat{F}_x$)} \\
    &= -\vec{v}(x) \cdot \vec{e}_x + \vec{v}'.
  \end{align*}
  
  The proof of~\ref{itm:flips:src:tgt} and \ref{itm:flips:other}, and
  the similar proof for the case where $\mat{A}$ is a pseudo-copy
  matrix, are analogous (see Appendix~\ref{sec:appendix}).
\end{proof}

\begin{prop}\label{prop:pspace:mone}
  $\zreach{\C}$ is \PSPACE-hard if $\C$ has a matrix with an entry
  equal to $-1$.
\end{prop}

\begin{proof}
  \newcommand{\init}{p_\text{init}}
  \newcommand{\acc}{p_\text{acc}}
  \newcommand{\dleft}{\textsc{Left}}
  \newcommand{\dright}{\textsc{Right}}

  We give a reduction, partially inspired by~\cite[Thm.~10]{BHM18},
  from the membership problem for linear bounded automata, which
  is \PSPACE-complete (\eg, see \cite[Sect.~9.3 and~13]{HU79}).

  Let $w \in {\{0, 1\}}^k$ and let $\A = (P, \Sigma, \delta, \init,
  \acc)$ be a linear bounded automaton where:
  \begin{itemize}
  \item $P$ is its finite set of control-states;

  \item $\Sigma = \{0, 1\}$ is its input and tape alphabet;

  \item $\delta \colon P \times \Sigma \to P \times \Sigma \times
    \{\dleft, \dright\}$ is its transition function; and

  \item $\init$ and $\acc$ are its initial and accepting
    control-states, respectively.
  \end{itemize}
  
  We construct an affine VASS $\V = (d, Q, T)$ and configurations
  $p(\vec{u})$, $q(\vec{v})$ such that $\V$ belongs to $\C$, and
  \[p(\vec{u}) \trans{*}_\Z q(\vec{v}) \iff \text{$\A$ accepts $w$}.\]

  For every control-state $p$ and head position $j$ of $\A$, there is
  a matching control-state in $\V$, \ie, $Q \defeq \{q_{p,j} : p \in
  P, 1 \leq j \leq k\} \cup \overline{Q}$, where $\overline{Q}$ will
  be auxiliary control-states. We associate two counters to each tape
  cell of $\A$, \ie, $d \defeq 2 \cdot k$. For readability, let us
  denote these counters $\{x_j, y_j : j \in [k]\}$.

  We represent the
  contents of tape cell $i$ by the sign of counter $y_j$, \ie, $y_j >
  0$ represents $0$, and $y_j < 0$ represents $1$. We will ensure that
  $y_j$ is never equal to $0$, which would otherwise be an undefined
  representation. Since $\V$ cannot directly test the sign of a
  counter, it will be possible for $\V$ to commit errors during the
  simulation of $\A$. However, we will construct $\V$ in such a way
  that erroneous simulations are detected.

  The gadget depicted in Figure~\ref{fig:gadget:trans} simulates a
  transition of $\A$ in three steps:
  \begin{itemize}
  \item $x_i$ is incremented;

  \item $y_i$ is incremented (resp.\ decremented) if the letter $a$ to
    be read is $0$ (resp.\ $1$);

  \item the sign of $y_i$ is flipped if the letter $b$ to be written
    differs from the letter $a$ to be read.
  \end{itemize}

  \begin{figure*}[h]
  \centering
  \begin{tikzpicture}[->, node distance=3.5cm, auto, very thick, scale=0.8, transform shape, font=\Large]
    %% Control-states
    \tikzset{every state/.style={inner sep=1pt, minimum size=30pt}}

    \node[state] (qi)               {$q_{p, i}$};
    \node[state] (a1) [right=of qi] {};
    \node[state] (a2) [right=of a1] {};
    \node[state] (qj) [right=of a2] {$q_{p', i+1}$};

    %% Transitions
    \path[->, font=\normalsize]
    (qi) edge node {$x_i \leftarrow x_i + 1$}              (a1)
    (a1) edge node {$y_i \leftarrow y_i + (-1)^a$}         (a2)
    (a2) edge node {$y_i \leftarrow (-1)^{b-a} \cdot y_i$} (qj)
    ;
  \end{tikzpicture}
  \caption{Gadget of $\V$ simulating transition $\delta(p, a) = (p',
    b, \dright)$ of $\mathcal{A}$. The gadget for direction $\dleft$
    is the same except for $q_{p', i+1}$ which is replaced by $q_{p',
      i-1}$. Note that $a$ and $b$ are fixed, hence expressions such
    as $(-1)^a$ are constants; they do not require
    exponentiation.}\label{fig:gadget:trans}
\end{figure*}
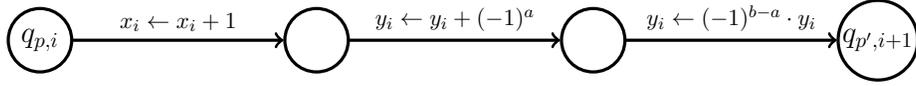

  Let $\vec{u} \in \Z^d$ be the vector such that for every $j \in
  [k]$: \[\vec{u}(x_j) \defeq 1 \text{ and }\vec{u}(y_j) \defeq
  (-1)^{w_j}.\] Provided that $\V$ starts in vector $\vec{u}$, we
  claim that: \begin{itemize} \item $\bigwedge_{j=1}^ k (|x_j| \geq
    |y_j| > 0)$ is an invariant;

  \item $\V$ has faithfully simulated $\A$ so far if and only if
    $\bigwedge_{j=1}^k (|x_j| = |y_j|)$ holds;

  \item if $\V$ has faithfully simulated $\A$ so far, then the sign of
    $y_j$ represents $w_j$ for every $j \in [k]$.
  \end{itemize}

  Let us see why this claim holds. Let $i \in [k]$. Initially, we have
  $|x_i| = |y_i|$ and the sign of $y_i$ set correctly. Assume we
  execute the gadget of Figure~\ref{fig:gadget:trans}, resulting in
  new values $x_i'$ and $y_i'$. Let $\lambda \geq 0$ be such that
  $|x_i| = |y_i| + \lambda$. Let $c \in \{0, 1\}$ be the letter
  represented by $y_i$. If $c = a$, then $|x_i'| = |y_i'| + \lambda$
  and the sign of $y_i'$ represents $b$ as desired. If $c \neq a$,
  then $|x_i'| = |y_i'| + (\lambda + 1)$. Thus, we have $|x_i'| =
  |y_i'|$ if and only if no error was made before and during the
  execution of the gadget.

  From the above observations, we conclude that $\A$ accepts $w$ if
  and only if there exist $i \in [k]$ and $\vec{v} \in \Z^d$ such that
  \[q_{\init,1}(\vec{u}) \trans{*}_\Z q_{\acc,i}(\vec{v})\] and
  $|\vec{v}(x_j)| = |\vec{v}(y_j)|$ for every $j \in[k]$. This can be
  tested using the gadget depicted in Figure~\ref{fig:gadget:acc}, which:
  \begin{itemize}
  \item detects nondeterministically that some control-state of the
    form $q_{\acc,i}$ has been reached;

  \item attempts to set $y_j$ to its absolute value for every $j \in
    [k]$;

  \item decrements $x_j$ and $y_j$ simultaneously for every $j \in
    [k]$.
  \end{itemize}

  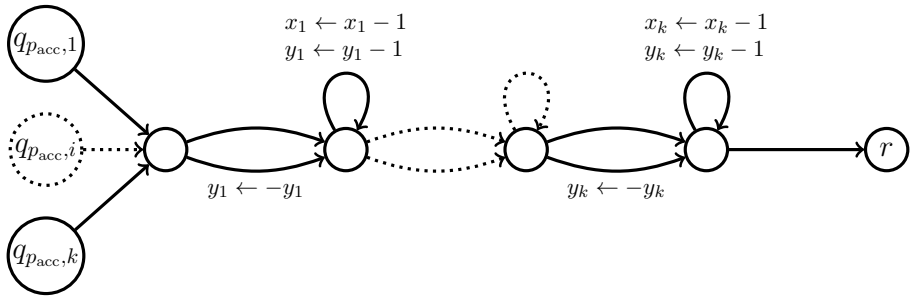
\begin{figure}[h]
  \centering
  \begin{tikzpicture}[->, node distance=2.25cm, auto, very thick, scale=0.8, transform shape, font=\Large]
    %% Control-states
    \tikzset{every state/.style={inner sep=1pt, minimum size=20pt}}

    \node[state] (r1)               {};
    \node[state] (r2) [right=of r1] {};
    \node[state] (r3) [right=of r2] {};
    \node[state] (r4) [right=of r3] {};
    \node[state] (r)  [right=of r4] {$r$};

    \node[state, dotted] (qi) [left=   1cm of r1] {$q_{\acc,i}$};
    \node[state]         (q1) [above=0.5cm of qi] {$q_{\acc,1}$};
    \node[state]         (qk) [below=0.5cm of qi] {$q_{\acc,k}$};

    %% Transitions
    \path[->, font=\normalsize]
    % accepting
    (q1) edge node {} (r1)
    (qk) edge node {} (r1)      
    
    % x_1, y_1
    (r1) edge[bend left=20]  node[]     {}                      (r2)
    (r1) edge[bend right=20] node[swap] {$y_1 \leftarrow -y_1$} (r2)
    (r2) edge[out=120, in=60, looseness=10] node[] {
      \begin{tabular}{l}
        $x_1 \leftarrow x_1 - 1$ \\
        $y_1 \leftarrow y_1 - 1$
      \end{tabular}
    } (r2)

    % x_k, y_k
    (r3) edge[bend left=20]  node[]     {}                      (r4)
    (r3) edge[bend right=20] node[swap] {$y_k \leftarrow -y_k$} (r4)
    (r4) edge[out=120, in=60, looseness=10] node[] {
      \begin{tabular}{l}
        $x_k \leftarrow x_k - 1$ \\
        $y_k \leftarrow y_k - 1$
      \end{tabular}
    } (r4)

    % done
    (r4) edge node {} (r)
    ;
    
    \path[->, font=\normalsize, dotted]
    % accepting
    (qi) edge node {} (r1)

    % x_i, y_i
    (r2) edge[bend left=20]  node[]     {} (r3)
    (r2) edge[bend right=20] node[swap] {} (r3)
    (r3) edge[out=120, in=60, looseness=10] node[] {} (r3)
    ;
  \end{tikzpicture}
  \caption{Gadget of $\V$ for tesing whether $\A$ was faithfully
    simulated and has accepted $w$.}\label{fig:gadget:acc}
\end{figure}

  Due to the above observations, it is \emph{only} possible to reach
  $r(\vec{0})$ if $|x_j| = |y_j|$ for every $j \in [k]$ before
  entering the gadget of Figure~\ref{fig:gadget:acc}. Thus, we are done
  proving the reduction since $\A$ accepts $w$ if and only if
  \[q_{\init,1}(\vec{u}) \trans{*}_\Z r(\vec{0}).\]

  \parag{Sign flips} The above construction considers sign flips as a
  ``native'' operation. However, this is not necessarily the case, and
  instead relies on the fact that class $\C$ either $0$-implements or
  $?$-implements sign flips, by
  Proposition~\ref{prop:exist:flip}. Thus, the reachability question
  must be changed to \[q_{\init,1}(\vec{u}, \vec{0}) \trans{*}_\Z
  r(\vec{0}, \vec{0})\] to take auxiliary counters into
  account. Moreover, if $\C$ $?$-implements sign flips, then extra
  transitions $(r, \mat{I}, \vec{e}_j, r)$ and $(r, \mat{I},
  -\vec{e}_j, r)$ must be added to $T$, for every auxiliary counter
  $j$, to allow counter $j$ to be set back to $0$. Note that
  control-state $r$ can only be reached \emph{after} the simulation of
  $\A$, hence it plays no role in the emulation of sign
  flips. Moreover, if there is an error during the simulation of $\A$
  and the extra transitions set the auxiliary counters to zero, we
  will stil detect it as the configuration will be of the form
  $r(\vec{w}, \vec{0})$ where $\vec{w} \neq \vec{0}$.
\end{proof}

In the two forthcoming propositions, we prove \PSPACE-hardness of the
remaining case.

\begin{prop}\label{prop:exist:swaps}
  If $\C$ contains a matrix with entries from $\{0, 1\}$ and a nonzero
  entry outside of its main diagonal, then it implements swaps,
  \ie, the operation $f \colon \Z^2 \to \Z^2$ such that $f(x, y)
  \defeq (y, x)$.
\end{prop}

\begin{proof}
  Let $n \geq 2$ and let $\mat{A} \in \C$ be a matrix with entries
  from $\{0, 1\}$ and a nonzero entry outside of its main
  diagonal. Let $d \defeq \size{\mat{A}}$. There exist $a, b \in [d]$
  such that $\mat{A}_{a, b} = 1$ and $a \neq b$. Let us extend
  $\mat{A}$ with $n + 1$ counters $X' \defeq X \cup \{z\}$ where
  $X \defeq \{x_i : i \in [n]\}$. More formally, let
  $\mat{A}' \defeq \ext{\mat{A}}{n + 1}$, $X' = [d + 1, d']$ and
  $d' \defeq d + n + 1$.

  For all $s, t \in X'$ such that $s \neq t$, let $\mat{B}_{s, t}
  \defeq \pi_{s, t}(\mat{A}')$ where $\pi_{s, t} \defeq (b; s) (a;
  t)$. For every distinct counters $x, y \in X$, let
  \[\mat{F}_{x, y} \defeq \mat{B}_{z, x} \cdot \mat{B}_{x, y} \cdot
  \mat{B}_{y, z}.\]

  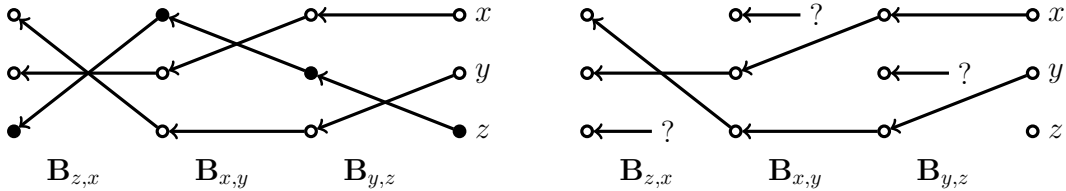
\begin{figure*}[h]
  \hfill
  %% Pseudo-transfer
  \begin{tikzpicture}[->, node distance=0.75cm, auto, very thick, scale=0.8, transform shape, font=\Large]
    %% Counters
    \tikzset{every state/.style={inner sep=1pt, minimum size=5pt}}
    \newcommand{\cspace}{2.25cm}

    \node[state, label={right:$x$}]       (xi0)                {};
    \node[state, label={right:$y$}]       (yi0) [below=of xi0] {};
    \node[state, label={right:$z$}, fill] (zi0) [below=of yi0] {};

    \node[state]       (xi1) [left=\cspace of xi0] {};
    \node[state, fill] (yi1) [below=       of xi1] {};
    \node[state]       (zi1) [below=       of yi1] {};

    \node[state, fill] (xi2) [left=\cspace of xi1] {};
    \node[state]       (yi2) [below=       of xi2] {};
    \node[state]       (zi2) [below=       of yi2] {};

    \node[state]       (xi3) [left=\cspace of xi2] {};
    \node[state]       (yi3) [below=       of xi3] {};
    \node[state, fill] (zi3) [below=       of yi3] {};

    %% Opérations
    \path[->]
    (xi0) edge node {} (xi1)
    (yi0) edge node {} (zi1)
    (zi0) edge node {} (yi1)

    (xi1) edge node {} (yi2)
    (yi1) edge node {} (xi2)
    (zi1) edge node {} (zi2)

    (xi2) edge node {} (zi3)
    (yi2) edge node {} (yi3)
    (zi2) edge node {} (xi3)
    ;

    \node[right={\dimexpr\cspace/8} of zi1, yshift=-20pt]
         {$\mat{B}_{y, z}$};
    \node[right={\dimexpr\cspace/8} of zi2, yshift=-20pt]
         {$\mat{B}_{x, y}$};
    \node[right={\dimexpr\cspace/8} of zi3, yshift=-20pt]
         {$\mat{B}_{z, x}$};
  \end{tikzpicture}%
  \hfill\hfill
  %% Pseudo-copy
  \begin{tikzpicture}[->, node distance=0.75cm, auto, very thick, scale=0.8, transform shape, font=\Large]
    %% Counters
    \tikzset{every state/.style={inner sep=1pt, minimum size=5pt}}
    \newcommand{\cspace}{2.25cm}

    \node[state, label={right:$x$}] (xi0)                {};
    \node[state, label={right:$y$}] (yi0) [below=of xi0] {};
    \node[state, label={right:$z$}] (zi0) [below=of yi0] {};

    \node[state] (xi1) [left=\cspace of xi0] {};
    \node[state] (yi1) [below=       of xi1] {};
    \node[state] (zi1) [below=       of yi1] {};

    \node[state] (xi2) [left=\cspace of xi1] {};
    \node[state] (yi2) [below=       of xi2] {};
    \node[state] (zi2) [below=       of yi2] {};

    \node[state] (xi3) [left=\cspace of xi2] {};
    \node[state] (yi3) [below=       of xi3] {};
    \node[state] (zi3) [below=       of yi3] {};

    \node (zr0) [left={\dimexpr\cspace/3} of yi0] {?};
    \node (zr1) [left={\dimexpr\cspace/3} of xi1] {?};
    \node (zr2) [left={\dimexpr\cspace/3} of zi2] {?};

    %% Opérations
    \path[->, font=\normalsize]
    (xi0) edge node {} (xi1)
    (yi0) edge node {} (zi1)
    (zr0) edge node {} (yi1)

    (xi1) edge node {} (yi2)
    (zi1) edge node {} (zi2)
    (zr1) edge node {} (xi2)

    (yi2) edge node {} (yi3)
    (zi2) edge node {} (xi3)
    (zr2) edge node {} (zi3)
    ;

    \node[right={\dimexpr\cspace/8} of zi1, yshift=-20pt]
         {$\mat{B}_{y, z}$};
    \node[right={\dimexpr\cspace/8} of zi2, yshift=-20pt]
         {$\mat{B}_{x, y}$};
    \node[right={\dimexpr\cspace/8} of zi3, yshift=-20pt]
         {$\mat{B}_{z, x}$};
  \end{tikzpicture}%
  \hfill\hfill
  \caption{Effect of applying $\mat{F}_{x, y}$, where the left
    (resp.\ right) diagram depicts the case where $\mat{A}$ is a
    transfer (resp.\ copy) matrix. An edge from counter $s$ to counter
    $t$ represents operation $s \leftarrow t$. Filled nodes indicate
    counters that necessarily hold $0$. Symbol ``?'' stands for an
    integer whose value is irrelevant and depends on $\mat{A}$ and the
    counter values.}\label{fig:swap}
\end{figure*}

  Intuitively, $\mat{B}_{s, t}$ moves the contents from some source
  counter $s$ to some target counter $t$, and $\mat{F}_{x, y}$
  implements a swap in three steps using an auxiliary counter $z$ as
  depicted in Figure~\ref{fig:swap}. In the case where $\mat{A}$ is a
  transfer matrix, $\mat{B}_{s, t}$ resets $s$, provided that $t$ held
  value~$0$.

  Let us consider the case where $\mat{A}$ is a transfer matrix.  From
  the definition of $\mat{B}_{s, t}$, it can be shown that for every
  $s, t, u \in X'$ such that $s \neq t$ and $u \not\in \{s, t\}$, the
  following holds:
  \begin{enumerate}[(i)]
  \item $\mat{B}_{s, t} \cdot \vec{e}_s = \vec{e}_t$,
    and\label{itm:src:tgt}
    
  \item $\mat{B}_{s, t} \cdot \vec{e}_u =
    \vec{e}_u$.\label{itm:other}
  \end{enumerate}

  Let us show that we $0$-implement swaps, so let \[V_X \defeq
  \left\{\vec{v} \in \Z^{d'} : \bigwedge_{j \not\in X}
  \vec{v}(j)=0\right\}.\] Let $\vec{v} \in V_X$ and let $x, y \in X$
  be such that $x \neq y$. By definition of $V_X$, $\vec{v} = \sum_{j \in X}
  \vec{v}(j) \cdot \vec{e}_j$.~Let \[\vec{v}'
  \defeq \sum_{j \in X \setminus \{x, y\}} \vec{v}(j) \cdot
  \vec{e}_j.\] Items~\ref{itm:op:f}, \ref{itm:op:untouched}
  and~\ref{itm:op:closure} of Definition~\ref{def:op:impl} are satisfied
  since we obtain the following by applications of \ref{itm:src:tgt}
  and~\ref{itm:other}:
  \begin{align*}
    \mat{F}_{x, y} \cdot \vec{v}
    &=
    \mat{F}_{x, y} \cdot \left(\vec{v}(x) \cdot \vec{e}_x +
    \vec{v}(y) \cdot \vec{e}_y\right) + \mat{F}_{x, y} \cdot \vec{v}'
    && \text{(by def.\ of $\vec{v}'$)} \\
    &= \mat{F}_{x, y} \cdot \left(\vec{v}(x) \cdot \vec{e}_x +
    \vec{v}(y) \cdot \vec{e}_y\right) + \vec{v}'
    && \text{(by~\ref{itm:other} and def.\ of $\mat{F}_{x, y}$)} \\
    &= \mat{B}_{z, x} \cdot \mat{B}_{x, y} \cdot \mat{B}_{y, z} \cdot
    (\vec{v}(x) \cdot \vec{e}_x + \vec{v}(y)
    \cdot \vec{e}_y) + \vec{v}'
    && \text{(by def.\ of $\mat{F}_{x, y}$)} \\
    &= \mat{B}_{z, x} \cdot \mat{B}_{x, y} \cdot (\vec{v}(x)
    \cdot \vec{e}_x + \vec{v}(y) \cdot \vec{e}_z) + \vec{v}'
    && \text{(by~\ref{itm:other} and~\ref{itm:src:tgt})} \\
    &= \mat{B}_{z, x} \cdot (\vec{v}(x) \cdot \vec{e}_y +
    \vec{v}(y) \cdot \vec{e}_z) + \vec{v}'
    && \text{(by~\ref{itm:src:tgt} and~\ref{itm:other})} \\
    &= \vec{v}(x) \cdot \vec{e}_y + \vec{v}(y) \cdot
    \vec{e}_x + \vec{v}'
    && \text{(by~\ref{itm:other} and~\ref{itm:src:tgt})}.
  \end{align*}

  The proof of~\ref{itm:src:tgt} and~\ref{itm:other}, and the
  similar proof for the case where $\mat{A}$ is a copy matrix, are
  analogous (see Appendix~\ref{sec:appendix}).
\end{proof}

\begin{prop}\label{prop:pspace:nonreset}
  $\zreach{\C}$ is \PSPACE-hard if $\C$ contains a matrix with entries
  from $\{0, 1\}$ and a nonzero entry outside of its main diagonal.
\end{prop}

\begin{proof}
  It is shown in~\cite{BHMR19} that $\Z$-reachability is \PSPACE-hard
  for affine VASS with swaps, using a reduction from the membership
  problem for linear bounded automata.

  Here, we may not have swaps as a ``native'' operation. However, by
  Proposition~\ref{prop:exist:swaps}, class $\C$ implements
  swaps. Thus, as in the proof of Proposition~\ref{prop:pspace:mone},
  if the reachability question is of the
  form \[p(\vec{u}) \trans{*}_\Z q(\vec{v}),\] then it must be changed
  to \[p(\vec{u}, \vec{0}) \trans{*}_\Z q(\vec{v}, \vec{0}).\]
  Moreover, if the class $\C$ $?$-implements swaps, then new
  transitions must be introduced to allow auxiliary counters to be set
  back to $0$. Recall that under $?$-implementation, there is no
  requirement on the value of the auxiliary counters, hence these new
  transitions do not interfere with the emulation of swaps.
\end{proof}

We now proceed to prove the main result of this subsection, namely
Theorem~\ref{thm:trichotomy}~\ref{itm:pspace}:

\begin{proof}[Proof of Theorem~\ref{thm:trichotomy}~\ref{itm:pspace}]
  Let $\monoid_k \defeq \C \cap \Z^{k \times k}$ for every $k \geq 1$.
  Theorem~7 of~\cite{BHM18} shows that $\zreach{\C}$ belongs
  to \PSPACE\ if each $\monoid_k$ is a finite monoid of at most
  exponential norm and size in $k$. Let us show that this is the
  case. First, since $\C$ is a class that contains only
  pseudo-transfer (reps.\ pseudo-copy) matrices, and since the product
  of two such matrices remains so, $\monoid_k$ is a monoid which is
  finite as $\monoid_k \subseteq \{-1, 0, 1\}^{k \times k}$. Moreover,
  by definitions of pseudo-transfer and pseudo-copy matrices, each
  such matrix can be described by cutting it into $k$ lines and
  specifying for each line either the position of the unique nonzero
  entry (which is $-1$ or $1$), or the lack of such entry. Therefore,
  for every $k \geq 1$, it is the case that $\norm{\monoid_k} \leq 1$
  and
  \begin{align*} |\monoid_k|
    &\leq (2k + 1)^k
    \leq (4k)^k
    = 2^{2k + k \log k}
    \leq 2^{\mathrm{poly}(k)}.
  \end{align*}

  It remains to show \PSPACE-hardness. By assumption, $\C$ contains a
  nonreset matrix $\mat{A}$. Since $\norm{\C} \leq 1$, we have
  $\norm{\mat{A}} = 1$ as no class can be such that $\norm{\C} =
  0$. If $\mat{A}$ contains an entry equal to $-1$, then we are done
  by Proposition~\ref{prop:pspace:mone}. Otherwise, $\mat{A}$ only has
  entries from $\{0, 1\}$, and hence we are done by
  Proposition~\ref{prop:pspace:nonreset}.
\end{proof}

\subsection{Undecidability}

In this subsection, we first show that any class $\C$, that does not
satisfy the requirements for $\zreach{\C} \in
\{\text{\NP-complete}, \allowbreak \text{\PSPACE-complete}\}$, must be
such that $\norm{\C} \geq 2$. We then show that this is sufficient to
mimic doubling, \ie, the operation $x \mapsto 2x$, even if $\C$ does
not contain a doubling matrix. In more details, we will (a)~construct
a matrix $\mat{C}$ that provides a sufficiently fast growth; which
will (b)~allow us to derive undecidability by revisiting a reduction
from the Post correspondence problem which depends on doubling.

\begin{prop}\label{prop:row:col:two}
  Let $\C$ be a class that contains some matrices $\mat{A}$ and
  $\mat{B}$ which are respectively not pseudo-copy and pseudo-transfer
  matrices. It is the case that $\norm{\C} \geq 2$.
\end{prop}

\begin{proof}
  By assumption, $\mat{A}$ and $\mat{B}$ respectively have a row and a
  column with at least two nonzero entries. We make use of the
  following lemma shown in Appendix~\ref{sec:appendix}:\smallskip
  \begin{quote}
    if $\C$ contains a matrix which has a row (resp.\ column) with at
    least two nonzero entries, then $\C$ also contains a matrix which
    has a row (resp.\ column) with at least two nonzero entries with
    the \emph{same sign}.
  \end{quote}\smallskip
  Since $\C$ is a class, we can assume that $\size{\mat{A}} =
  \size{\mat{B}} = d$ for some $d \geq 2$, as otherwise the smallest
  matrix can be enlarged. Thus, there exist $i, i', j, k \in [d]$ and
  $a, b, a', b' \neq 0$ such that:
  \begin{itemize}
  \item $j \neq k$;

  \item $\mat{A}_{i, j} = a$ and $\mat{A}_{i, k} = b$;

  \item $\mat{B}_{j, i'} = a'$ and $\mat{B}_{k, i'} = b'$; and
    
  \item $a > 0 \iff b > 0$ and $a' > 0 \iff b' > 0$.
  \end{itemize}
  Note that the reason we can assume $\mat{A}$ and $\mat{B}$ to share
  counters $j$ and $k$ is due to $\C$ being closed under counter
  renaming.

  We wish to obtain a matrix with entry \[\mat{A}_{i, j} \cdot
  \mat{B}_{j, i'} + \mat{A}_{i, k} \cdot \mat{B}_{k, i'} = a \cdot a'
  + b \cdot b'.\] We cannot simply pick $\mat{A} \cdot \mat{B}$ as
  $(\mat{A} \cdot \mat{B})_{i, i'}$ may differ from this value due to
  other nonzero entries. Hence, we rename all counters of $\mat{B}$,
  except for $j$ and $k$, with fresh counters. This way, we avoid
  possible overlaps and we can select precisely the four desired
  entries. More formally, let $\mat{A}' \defeq \ext{\mat{A}}{d}$,
  $\mat{B}' \defeq \ext{\mat{B}}{d}$ and $\mat{C} \defeq \mat{A}'
  \cdot \perm{\mat{B}'}{\sigma}$, where $\sigma \colon [2d] \to [2d]$
  is the permutation: \[ \sigma \defeq \prod_{\ell \in [d] \setminus
    \{j, k\}} (\ell; \ell + d).  \] Let $i'' \defeq i' + d$ if $i'
  \not\in \{j, k\}$ and $i'' \defeq i'$ otherwise. We have:
  \begin{align*}
    \mat{C}_{i, i''}
    &= \sum_{\ell \in [2d]} \mat{A}'_{i, \ell} \cdot
    \mat{B}'_{\sigma(\ell), i'}    
    && \text{(by def.\ of $\mat{C}$ and by $\sigma(i'') = i'$)} \\
    &= \sum_{\ell \in [d] \text{ s.t. } \sigma(\ell) \in [d]}
    \mat{A}'_{i, \ell} \cdot \mat{B}'_{\sigma(\ell), i'}
    && \text{(by def.\ of $\mat{A}'$ and $\mat{B}'$, and by $i, i' \in
      [d]$)} \\
    &= \mat{A}'_{i, j} \cdot \mat{B}'_{j, i'} + \mat{A}'_{i, k} \cdot
    \mat{B}'_{k, i'} \\
    &= a \cdot a' + b \cdot b'.
  \end{align*}
  
  Since $a$ and $b$ (resp.\ $a'$ and $b'$) have the same sign, and
  since $a, b, a', b' \neq 0$, we conclude that $|\mat{C}_{i, i''}|
  \geq 2$ and consequently that $\norm{\C} \geq \norm{\mat{C}} \geq
  2$.
\end{proof}

To avoid cumbersome subscripts, we write $\vec{e}$ for $\vec{e}_1$ in
the rest of the section. Moreover, let $\lambda_\ell(\mat{C}) \defeq
(\mat{C}^\ell \cdot \vec{e})(1)$ for every matrix $\mat{C}$ and
$\ell \in \N$.

The following technical lemma will be key to mimic doubling. It shows
that, from any class of norm at least $2$, we can extract a matrix
with sufficiently fast growth.

\begin{lem}\label{lem:doubling}
  For every class of matrices $\C$ such that $\norm{\C} \geq 2$, there
  exists $\mat{C} \in \C$ with $\lambda_{n+1}(\mat{C}) \geq
  2 \cdot \lambda_n(\mat{C})$ for every $n \in \N$.
\end{lem}

\begin{proof}
  Let $\mat{A} \in \C$ be a matrix with some entry $c$ such that $|c|
  \geq 2$. We can assume that $c \geq 2$. Indeed, if it is negative,
  then we can multiply $\mat{A}$ by a suitable permutation of itself
  to obtain an entry equal to $c \cdot c$. We can further assume that
  $c$ is the largest positive coefficient occurring within $\mat{A}$,
  and that it lies on the first column of $\mat{A}$, \ie, $\mat{A}_{k,
    1} = c$ for some $k \in [d]$ where $d \defeq \size{\mat{A}}$. We
  consider the case where $k = 1$. The case where $k \neq 1$ will be
  discussed later.

  For readability, we rename counters $\{1, 2, \ldots, d\}$
  respectively by $X \defeq \{x_1, x_2, \ldots, x_d\}$. Note that
  $(\mat{A} \cdot \vec{e})(x_1) = c \geq 2 \cdot \vec{e}(x_1)$ as
  desired. However, vector $\mat{A} \cdot \vec{e}$ may now hold
  nonzero values in counters $x_2, \ldots, x_d$. Therefore, if we
  multiply this vector by $\mat{A}$, some ``noise'' will be added to
  counter $x_1$. If this noise is too large, then it may cancel the
  growth of $x_1$ by $\approx c$. We address this issue by introducing
  extra auxiliary counters replacing $x_2, \ldots, x_d$ at each
  ``iteration''. Of course, we cannot have infinitely many auxiliary
  counters. Fortunately, after a sufficiently large number $m$ of
  iterations, the auxiliary counters used at the first iteration
  will contain sufficiently small noise so that the process can
  restart from there.

  More formally, let $\mat{A}' \defeq \ext{\mat{A}}{|Y|}$ where $Y
  \defeq \{y_{i, j} : 0 \leq i < m, j \in [2, d]\}$ is the set of
  auxiliary counters, and $m \geq 1$ is a sufficiently large constant
  whose value will be picked later. Let $V$ be the set of vectors
  $\vec{v} \in \Z^{|X| + |Y|}$ satisfying $\vec{v}(x_1) > 0$
  and \[|\vec{v}(y_{i, j})| \leq \left(\frac{3c}{4}\right)^i \cdot
  \frac{\vec{v}(x_1)}{4d} \text{ for every } y_{i, j} \in Y.\]

  Let us fix some vector $\vec{v}_0 \in V$. For every $0 \leq i < m$,
  let $\mat{B}_i \defeq \perm{\mat{A}'}{\sigma_i}$ and $\vec{v}_{i+1}
  \defeq \mat{B}_i \cdot \vec{v}_i$ where $\sigma_i$ is the
  permutation \[\sigma_i \defeq \prod_{j \in [2, d]} (x_j; y_{i,
    j}).\] We claim that: \[\vec{v}_m(x_1) \geq 2 \cdot \vec{v}_0(x_1)
  \text{ and } \vec{v}_m \in V.\] The validity of this claim proves
  the lemma. Indeed, $\mat{C} \cdot \vec{v}_0 = \vec{v}_m$ where
  $\mat{C} \defeq \mat{B}_{m - 1} \cdots \mat{B}_1 \cdot
  \mat{B}_0$. Hence, an application of $\mat{C}$ yields a vector whose
  first component has at least doubled. Since $\vec{e} \in V$ and the
  resulting vector also belong to $V$, this can be iterated
  arbitrarily many times.

  Let us first establish the following properties for every $0 \leq i
  < m$ and $j \in [2, d]$:
  \begin{enumerate}[(a)]
  \item $\vec{v}_i(y_{i, j}) = \vec{v}_{i-1}(y_{i, j}) = \cdots =
    \vec{v}_0(y_{i, j})$ and $\vec{v}_{i+1}(y_{i, j}) =
    \vec{v}_{i+2}(y_{i, j}) = \cdots = \vec{v}_m(y_{i,
      j})$;\label{itm:v:basic}
    
  \item $\vec{v}_{i+1}(x_1) \in \left[\frac{3c}{4} \cdot
    \vec{v}_i(x_1), \frac{5c}{4} \cdot \vec{v}_i(x_1)\right]$;
    and\label{itm:v:x1:bound}
    
  \item $|\vec{v}_{i+1}(y_{i, j})| \leq 2c \cdot
    \vec{v}_i(x_1)$.\label{itm:v:yij:bound}
  \end{enumerate} \smallskip
  Property~\ref{itm:v:basic}, which follows from the definition of
  $\mat{B}_i$, essentially states that the contents of counter $y_{i,
    j}$ is only altered from $\vec{v}_i$ to
  $\vec{v}_{i+1}$. Properties~\ref{itm:v:x1:bound}
  and~\ref{itm:v:yij:bound} bound the growth of the counters in terms
  of $x_1$. Let us prove these two latter properties by induction on
  $i$.

  By definition of $\vec{v}_{i+1}$, we have $\vec{v}_{i+1}(x_1) = c
  \cdot \vec{v}_i + \delta$ where \[\delta \defeq \sum_{\substack{z
      \neq x_1}} (\mat{B}_i)_{x_1, z} \cdot \vec{v}_i(z).\] Therefore,
  $\vec{v}_{i+1}(x_1) \in [c \cdot \vec{v}_i(x_1) - |\delta|, c \cdot
    \vec{v}_i(x_1) + |\delta|]$, and hence
  property~\ref{itm:v:x1:bound} follows from:
  \begin{align*}
    |\delta|
    &\leq \sum_{\substack{z \in X \cup Y \\ z \neq x_1}} |(\mat{B}_i)_{x_1,
      z}| \cdot |\vec{v}_i(z)| \\
    &= \sum_{j \in [2, d]} |\mat{A}'_{x_1, x_j}| \cdot |\vec{v}_i(y_{i,
      j})|
    && \text{(by def.\ of $\mat{B}_i$ and $\sigma_i$)} \\
    &= \sum_{j \in [2, d]} |\mat{A}'_{x_1, x_j}| \cdot |\vec{v}_0(y_{i, j})|
    && \text{(by~\ref{itm:v:basic})} \\
    &\leq dc \cdot (3c / 4)^i \cdot \frac{\vec{v}_0(x_1)}{4d}
    && \text{(by $\vec{v}_0 \in V$ and by maximality of $c$)} \\
    &\leq dc \cdot \frac{\vec{v}_i(x_1)}{4d}
    && \text{(by $\vec{v}_i(x_1) \geq (3c/4)^i \cdot \vec{v}_0(x_1)$
      from~\ref{itm:v:x1:bound})} \\
    &= \frac{c}{4} \cdot \vec{v}_i(x_1).
  \end{align*}

  Similarly, property~\ref{itm:v:yij:bound} holds since, for every
  $j \in [2, d]$:
  \begin{align*}
    |\vec{v}_{i+1}(y_{i, j})|
    &\leq |\mat{A}'_{x_j, x_1}| \cdot \vec{v}_i(x_1) + \sum_{\ell \in
      [2, d]} |\mat{A}'_{x_j, x_\ell}| \cdot |\vec{v}_i(y_{i, \ell})|
    && \text{(by def.\ of $\mat{B}_i$ and $\sigma_i$)} \\
    &\leq c \cdot \vec{v}_i(x_1) + dc \cdot (3c / 4)^i \cdot
    \frac{\vec{v}_0(x_1)}{4d}
    && \text{(by \ref{itm:v:basic} and $\vec{v}_0 \in V$)} \\
    &\leq c \cdot \vec{v}_i(x_1) + dc \cdot \frac{\vec{v}_i(x_1)}{4d}
    && \text{(by~\ref{itm:v:x1:bound})} \\
    &\leq 2c \cdot \vec{v}_i(x_1).
  \end{align*}

  We may now prove the claim. Let $m$ be sufficiently large so that
  $(3c/4)^m \geq 8cd$. We have $\vec{v}_m(x_1) \geq (3c/4)^m \cdot
  \vec{v}_0(x_1) \geq 8cd \cdot \vec{v}_0(x_1)$
  by~\ref{itm:v:x1:bound} and definition of $m$. Hence, since $c
  \geq 2$ and $d \geq 1$, we have $\vec{v}_m(x_1) \geq 2 \cdot
  \vec{v}_0(x_1)$, which satisfies the first part of the
  claim. Moreover, the second part of the claim, namely $\vec{v}_m \in
  V$, holds since for every $y_{i, j} \in Y$, we have:
  \begin{align*}
    |\vec{v}_m(y_{i, j})|
    &= |\vec{v}_{i+1}(y_{i, j})|
    && \text{(by~\ref{itm:v:basic})} \\
    &\leq 2c \cdot \vec{v}_i(x_1)
    && \text{(by~\ref{itm:v:yij:bound})} \\
    &\leq 2c \cdot \frac{\vec{v}_m(x_1)}{(3c/4)^{m - i}}
    && \text{(by $\vec{v}_m(x_1) \geq (3c/4)^{m - i} \cdot
      \vec{v}_i(x_1)$ from~\ref{itm:v:x1:bound})} \\    
    &= 2c \cdot (3c/4)^i \cdot \frac{\vec{v}_m(x_1)}{(3c/4)^m} \\
    &\leq 2c \cdot (3c/4)^i \cdot \frac{\vec{v}_m(x_1)}{8cd}
    && \text{(by def.\ of $m$)} \\
    &= (3c/4)^i \cdot \frac{\vec{v}_m(x_1)}{4d}.
  \end{align*}

  We are done proving the lemma for the case $\mat{A}_{k, 1} = c \geq
  2$ with $k = 1$. This case is slightly simpler as $c$ lies on the
  main diagonal of $\mat{A}$ which means that $\vec{v}_{i+1}(x_1)
  \approx c \cdot \vec{v}_i(x_1)$. If $k \neq 1$, then we have
  $\vec{v}_{i+1}(x_k) \approx c \cdot \vec{v}_i(x_1)$ instead, which
  breaks composability for the next iteration. However, this is easily
  fixed by swapping the names of counters $x_k$ and~$x_1$.
\end{proof}

Let us fix a class $\C$ such that $\norm{\C} \geq 2$ and the matrix
$\mat{C}$ obtained for $\C$ from Lemma~\ref{lem:doubling}. For
simplicity, we will write $\lambda_\ell$ instead of
$\lambda_\ell(\mat{C})$. We prove two intermediary propositions that
essentially show that $\mat{C}$ can encode binary strings. Let
$f_b(\vec{v}) \defeq \mat{C} \cdot \vec{v} + b \cdot \vec{e}$ for both
$b \in \{0, 1\}$ and every $\vec{v} \in \Z^{\size{\mat{C}}}$. Let
$f_\varepsilon$ be the identity function, and let $f_x \defeq f_{x_n}
\circ \cdots \circ f_{x_2} \circ f_{x_1}$ for every $x \in {\{0,
1\}}^n$. Let $\gamma_x \defeq f_x(\vec{e})(1)$ for every $x \in {\{0,
1\}}^*$. Let $\supp{\varepsilon} \defeq \emptyset$ and $\supp{w}
\defeq \{i \in [k] : w_i = 1\}$ be the ``support'' of $w$ for every
sequence $w \in \{0, 1\}^+$ of length $k > 0$.

\begin{prop}\label{prop:gamma:val}
  For every $x \in {\{0, 1\}}^*$, the following holds: $\gamma_x =
  \lambda_{|x|} + \sum_{i \in \supp{x}} \lambda_{|x|-i}$.
\end{prop}

\begin{proof}
  It suffices to show that $f_x(\vec{e}) = \mat{C}^{|x|} \cdot \vec{e}
  + \sum_{i \in \supp{x}} \mat{C}^{|x|-i} \cdot \vec{e}$ for every $x
  \in {\{0, 1\}}^*$. Let us prove this by induction on $|x|$. If $|x| =
  0$, then $x = \varepsilon$, and hence $f_x(\vec{e}) = \vec{e} =
  \mat{C}^0 \cdot \vec{e}$. Assume that $|x| > 0$ and that the claim
  holds for sequences of length $|x| - 1$. There exist $b \in \{0,
  1\}$ and $w \in {\{0, 1\}}^*$ such that $x = wb$. We have:
  \begin{align*}
    f_x(\vec{e})
    &= \mat{C} \cdot f_w(\vec{e}) + b \cdot \vec{e} \\
    &= \mat{C} \cdot \left(\mat{C}^{|w|} \cdot \vec{e} + \sum_{i \in
      \supp{w}} \mat{C}^{|w|-i} \cdot \vec{e} \right) + b \cdot
    \vec{e}
    && \text{(by induction hypothesis)} \\
    &= \left(\mat{C}^{|w|+1} \cdot \vec{e} + \sum_{i \in \supp{w}}
    \mat{C}^{|w|+1-i} \cdot \vec{e}\right) + b \cdot \vec{e} \\
    &= \left(\mat{C}^{|x|} \cdot \vec{e} + \sum_{i \in \supp{x}
      \setminus \{|x|\}} \mat{C}^{|x|-i} \cdot \vec{e}\right) +
    b \cdot \mat{C}^{|x|-|x|} \cdot \vec{e} \\
    &= \mat{C}^{|x|} \cdot \vec{e} + \sum_{i \in \supp{x}}
    \mat{C}^{|x|-i} \cdot \vec{e}
    && \text{(by def.\ of $\supp{x}$).} \qedhere
  \end{align*}
\end{proof}

\begin{prop}\label{prop:lex:val}
  For every $x, y \in {\{0, 1\}}^*$, it is the case that $x = y$ if and
  only if $\gamma_x = \gamma_y$.
\end{prop}

\begin{proof}
  Let $<_\text{lex}$ denote the lexicographical order over ${\{0,
  1\}}^*$. It is sufficient to show that for every $x, y \in {\{0,
  1\}}^*$ the following holds: if $x <_\text{lex} y$, then $\gamma_x <
  \gamma_y$. Indeed, if this claim holds, then for every $x, y \in
  {\{0, 1\}}^*$ such that $x \neq y$, we either have $x <_\text{lex} y$
  or $y <_\text{lex} x$, which implies $\gamma_x \neq \gamma_y$ in
  both cases.

  Let us prove the claim. Let $x, y \in {\{0, 1\}}^*$ be such that $x
  <_\text{lex} y$. We either have $|x| < |y|$ or $|x| = |y|$. If the
  former holds, then the claim follows from:
  \begin{align*}
    \gamma_x
    &= \lambda_{|x|} + \sum_{i \in \supp{x}} \lambda_{|x|-i}
    && \text{(by Proposition~\ref{prop:gamma:val})} \\
    &\leq \lambda_{|x|} + \sum_{i = 1}^{|x|} \lambda_{|x|} / 2^i
    && \text{(by Lemma~\ref{lem:doubling})} \\
    &= \lambda_{|x|} \cdot \left(1 + \sum_{i = 1}^{|x|} 1/2^i\right) \\
    &= \lambda_{|x|} \cdot (2 - 1/2^{|x|}) \\
    &< 2 \cdot \lambda_{|x|}
    && \text{(since $\lambda_{|x|} > 0$)} \\
    &\leq \lambda_{|y|}
    && \text{(by Lemma~\ref{lem:doubling})} \\
    &\leq \gamma_y
    && \text{(by Proposition~\ref{prop:gamma:val}).}
  \end{align*}       

  It remains to prove the case where $|x| = |y| = k$ for some $k >
  0$. Since $x <_\text{lex} y$, there exist $u, v, w \in {\{0, 1\}}^*$
  such that $x = u0v$ and $y = u1w$. Let $\ell \defeq k - |u| -
  1$. Note that $\ell = |v| = |w|$. The proof is completed by
  observing that:
  \begin{align*}
    \gamma_y - \gamma_x
    &= \lambda_\ell + \sum_{i \in \supp{w}}
    \lambda_{\ell - i} - \sum_{i \in \supp{v}}
    \lambda_{\ell - i}
    && \text{(by Proposition~\ref{prop:gamma:val})} \\
    &\geq \lambda_\ell - \sum_{i = 1}^{|v|}
    \lambda_{\ell - i} \\
    &\geq \lambda_\ell - \sum_{i = 1}^{|v|}
    \lambda_\ell / 2^i
    && \text{(by $\lambda_\ell \geq 2^i \cdot \lambda_{\ell - i}$ from
      Lemma~\ref{lem:doubling})} \\    
    &= \lambda_\ell \cdot \left(1 - \sum_{i = 1}^{\ell} 1 / 2^i\right)
    && \text{(by $|v| = \ell$)} \\
    &= \lambda_\ell / 2^\ell \\
    &> 0.
    && \qedhere
  \end{align*}
\end{proof}

We may finally prove the last part of our trichotomy.

\begin{thm}\label{thm:undec:normtwo}
  $\zreach{\C}$ is undecidable if $\norm{\C} \geq 2$.
\end{thm}

\begin{proof}
  \newcommand{\tile}[2]{%
    \left[
      \begin{tabular}{cc}
        \!\!\!{$#1$}\!\!\! \\
        \hline
        \!\!\!{$#2$}\!\!\!
      \end{tabular}
      \right]%
  }

  We give a reduction from the Post correspondence problem inspired
  by~\cite{Rei15}. There, counter values can be doubled as a
  ``native'' operation. Here, we adapt the construction with our
  emulation of doubling. Let us consider an instance of the Post
  correspondence problem over alphabet $\{0, 1\}$: \[\Gamma \defeq
  \left\{\tile{u_1}{v_1}, \tile{u_2}{v_2}, \ldots,
  \tile{u_\ell}{v_\ell}\right\}.\] We say that $\Gamma$ \emph{has a
    match} if there exists $w \in \Gamma^+$ such that the underlying
  top and bottom sequences of $w$ are equal.

  Let $\mat{C}$ be the matrix obtained for $\C$ from
  Lemma~\ref{lem:doubling}, let $d \defeq \size{\mat{C}}$, and let
  $\vec{e}$ be of size $d$. For every $x \in {\{0, 1\}}^*$, let $g_x$
  and $h_x$ be the linear mappings over $\Z^{2d}$ defined as $f_x$,
  but operating on counters $1, 2, \ldots, d$ and counters $d+1, d+2,
  \ldots, 2d$ respectively.

  \begin{figure}[h]
  \centering
  \begin{tikzpicture}[->, node distance=1.75cm, auto, very thick, scale=0.8, transform shape, font=\Large]      
    %% Control-states
    \tikzset{every state/.style={inner sep=1pt, minimum size=25pt}}

    \node[state]         (p)                   {$p$};
    \node[state]         (p1) [above=1cm of p] {$q_1$};
    \node[state, dotted] (pi) [left=     of p] {$q_i$};
    \node[state]         (pl) [below=1cm of p] {$q_\ell$};
    \node[state]         (q)  [right=5cm of p] {$r$};

    %% Transitions
    \path[->]
    % w1
    (p)  edge[bend left=20]  node[yshift=+5pt] {$g_{u_1}$} (p1)
    (p1) edge[bend left=20]  node[yshift=+5pt] {$h_{v_1}$} (p)

    % wl
    (p)  edge[bend right=20] node[yshift=-5pt, swap] {$g_{u_\ell}$} (pl)
    (pl) edge[bend right=20] node[yshift=-5pt, swap] {$h_{v_\ell}$} (p)

    % x = y?
    (p) edge[]           node[] {$(-\vec{e}, -\vec{e})$} (q)
    (q) edge[loop above] node[] {$(-\vec{e}, -\vec{e})$} ()    
    ;

    \path[dotted, ->]
    % wi
    (p)  edge[bend left=20] node[xshift=-5pt] {$g_{u_i}$} (pi)
    (pi) edge[bend left=20] node[xshift=-5pt] {$h_{v_i}$} (p)
    ;
  \end{tikzpicture}
  \caption{Affine VASS $\V$ for the Post correspondence problem. Arcs
    labeled by mappings of the form $g_x$ and $h_x$ stand for
    sequences of $|x|$ transitions implementing $g_x$ and
    $h_x$.}\label{fig:pcp}
\end{figure}
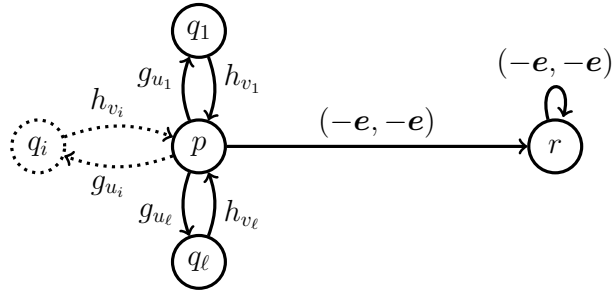

  Let $\V \defeq (2d, Q, T)$ be the affine VASS such that $Q$ and $T$
  are as depicted in Figure~\ref{fig:pcp}. Note that $\V$ belongs to
  $\C$. Indeed, $g_x$ and $h_x$ can be obtained from matrix $\mat{C}
  \in \C$ and the fact that $\C$ is a class, and hence closed under
  counter renaming and larger dimensions. We claim that \[p(\vec{e},
  \vec{e}) \trans{*}_\Z r(\vec{e}, \vec{e}) \text{ if and only if }
  \Gamma \text{ has a match}.\]
  
  Note that any sequence $w \in T^+$ from $p$ to $p$ computes
  $g_{w_x} \circ h_{w_y}$ for some
  word \[\tile{w_x}{w_y} \in \Gamma^+.\] Therefore:
  \begin{align*}
    &\phantom{{}\iff{}} p(\vec{e}, \vec{e}) \trans{*}_\Z r(\vec{e},
    \vec{e}) \\
    &\iff \exists w \in T^+, \vec{v} \in \Z^{2d} : p(\vec{e}, \vec{e})
    \trans{w}_\Z p(\vec{v}) \trans{*}_\Z r(\vec{e}, \vec{e}) \\
    &\iff \exists w \in T^+, \vec{v} \in \Z^{2d} :
    p(\vec{e}, \vec{e}) \trans{w}_\Z p(\vec{v}) \text{ and
    } \vec{v}(1) = \vec{v}(d + 1) \\
    &\iff \exists w \in T^+ : \gamma_{w_x} = \gamma_{w_y}
    && \text{(by def.\ of $g$, $h$ and $\gamma$)} \\
    &\iff \exists w \in T^+: w_x = w_y
    && \text{(by Proposition~\ref{prop:lex:val})} \\
    &\iff \Gamma \text{ has a match.}
    &&\qedhere
  \end{align*}
\end{proof}

We conclude this section by proving
Theorem~\ref{thm:trichotomy}~\ref{itm:undec} which can be
equivalently formulated as follows:

\begin{cor}
  $\zreach{\C}$ is undecidable if $\C$ does not only contain
  pseudo-transfer matrices and does not only contain pseudo-copy
  matrices.
\end{cor}

\begin{proof}
  By Proposition~\ref{prop:row:col:two}, $\norm{\C} \geq 2$, hence
  undecidability follows from Theorem~\ref{thm:undec:normtwo}.
\end{proof}

\section{A complexity dichotomy for reachability}
\label{sec:dichotomy}
This section is devoted to the following complexity dichotomy on
$\reach{\C}$, which is mostly proven by exploiting notions and results
from the previous section:

\begin{thm}\label{thm:dichotomy}
  The reachability problem $\reach{\C}$ is equivalent to the
  (standard) VASS reachability problem if $\C$ only contains
  permutation matrices, and is undecidable otherwise.
\end{thm}

\subsection{Decidability}
\newcommand{\I}{\mathcal{I}}

Note that the (standard) VASS reachability problem is the problem
$\reach{\mathcal{I}}$ where $\I \defeq \bigcup_{n \geq 1}
\mat{I}_n$. Clearly $\reach{\I} \leq \reach{\C}$ for any class
$\C$. Therefore, it suffices to show the following:

\begin{prop}
  $\reach{\C} \leq \reach{\I}$ for every $\C$ that only contains
  permutation matrices.
\end{prop}

\begin{proof}
  Let $\V = (d, Q, T)$ be an affine VASS that belongs to $\C$. We
  construct a (standard) VASS $\V' = (d, Q', T')$ that simulates
  $\V$. Recall that a (standard) VASS is an affine VASS that only uses
  the identity matrix. For readability, we omit the identity matrix on
  the transitions of $\V'$. We assume without loss of generality that
  each transition $t \in T$ satisfies either $\tvec{t} = \vec{0}$ or
  $\tmat{t} = \mat{I}$. Indeed, since permutation matrices are
  nonnegative, every transition of $T$ can be splitted in two parts:
  first applying its matrix, and then its vector.

  The control-states and transitions of $\V'$ are defined as
  $Q' \defeq \{q_\sigma : q \in Q, \sigma \in \S_d\}$ and $T' \defeq
  S_{\bm\circlearrowright} \cup S_\text{vec}$, which are to be defined
  shortly. Intuitively, each control-state of $\V'$ stores the current
  control-state of $\V$ together with the current renaming of its
  counters. Whenever a transition $t \in T$ such that $\tvec{t}
  = \vec{0}$ is to be applied, this means that the counters must be
  renamed by the permutation $\tmat{t}$. This is achieved
  by: \[S_{\bm\circlearrowright} \defeq \{(p_\sigma, \vec{0},
  q_{\pi \circ \sigma}) : (p, \mat{P}_\pi, \vec{0}, q) \in
  T, \sigma \in \S_d\}.\] Similarly, whenever a transition $t \in T$
  such that $\tmat{t} = \mat{I}$ is to be applied, this means that
  $\tvec{t}$ should be added to the counters, but in accordance to the
  current renaming of the counters. This is achieved by:
  \[S_\text{vec} \defeq \{(p_\sigma, \mat{P}_\sigma \cdot \vec{b},
  q_\sigma) : (p, \mat{I}, \vec{b}, q), \sigma \in \S_d\}.
    \]

  A routine induction shows that \[p(\vec{u}) \trans{*}_\N
  q(\vec{v}) \text{ in $\V$ iff } p_\varepsilon(\vec{u}) \trans{*}_\N
  q_\sigma(\mat{P}_\sigma \cdot \vec{v}) \text{ in $\V'$},\] where
  $\varepsilon$ denotes the identity permutation. Since this amounts
  to finitely many reachability queries, \ie, $|\S_d| = d!$ queries,
  this yields a Turing reduction\footnote{Although it is not necessary
  for our needs, the reduction can be made \emph{many-one} by weakly
  computing a matrix multiplication by $\mat{P}_{\sigma^{-1}}$ onto
  $d$ new counters, from each control-state $q_\sigma$ to a common
  state $r$.}.
\end{proof}

\subsection{Undecidability}

We show undecidability by considering three types of classes: (1)~classes
with matrices with negative entries; (2)~nontransfer \emph{and} noncopy
classes; and (3)~transfer
\emph{or} copy classes. In each case, we will argue that an ``undecidable
operation'' can be simulated, namely: zero-tests, doubling and resets.

\begin{prop}\label{prop:undec:neg}
  $\reach{\C}$ is undecidable for every class $\C$ that contains a
  matrix with some negative entry.
\end{prop}

\begin{proof}
  Let $\mat{A} \in \C$ be a matrix such that $\mat{A}_{i, j} < 0$ for
  some $i, j \in [d]$ where $d \defeq \size{\mat{A}}$. We show how a
  two counter Minsky machine $\mathcal{M}$ can be simulated by an
  affine VASS $\V$ belonging to $\C$. Note that we only have to show how to
  simulate zero-tests. The affine VASS $\V$ has $2d$ counters:
  counters $j$ and $j + d$ which represent the two counters $x$ and
  $y$ of $\mathcal{M}$; and $2d - 2$ auxiliary counters which will be
  permanently set to value $0$.

  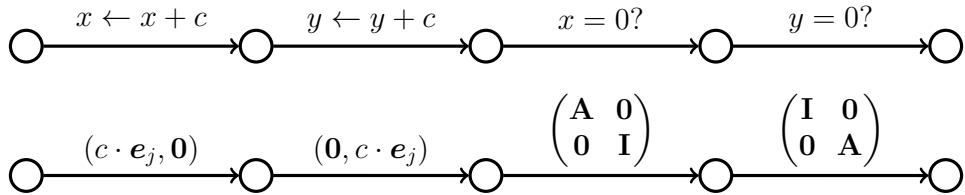
\begin{figure}[h!]
  \centering
  \begin{tikzpicture}[->, node distance=3.25cm, auto, very thick, scale=0.8, transform shape, font=\Large]
    %% Minsky machine
    % Control-states
    \tikzset{every state/.style={inner sep=1pt, minimum size=15pt}}

    \node[state] (q1)                {};
    \node[state] (q2) [right= of q1] {};
    \node[state] (q3) [right= of q2] {};
    \node[state] (q4) [right= of q3] {};
    \node[state] (q5) [right= of q4] {};

    % Transitions
    \path[->]
    (q1) edge[] node[yshift=4pt] {$x \leftarrow x + c$} (q2)
    (q2) edge[] node[yshift=2pt] {$y \leftarrow y + c$} (q3)
    (q3) edge[] node[yshift=4pt] {$x = 0$?}             (q4)
    (q4) edge[] node[yshift=2pt] {$y = 0$?}             (q5)
    ;
    
    %% Affine VASS
    % Control-states
    \node[state] (r1) [below=45pt of q1] {};
    \node[state] (r2) [right= of r1] {};
    \node[state] (r3) [right= of r2] {};
    \node[state] (r4) [right= of r3] {};
    \node[state] (r5) [right= of r4] {};

    % Transitions
    \path[->]
    (r1) edge[] node[] {$(c \cdot \vec{e}_j, \vec{0})$} (r2)
    (r2) edge[] node[] {$(\vec{0}, c \cdot \vec{e}_j)$} (r3)
    (r3) edge[] node[] {
      $
      \begin{pmatrix}
        \mat{A} & \mat{0} \\
        \mat{0} & \mat{I}
      \end{pmatrix}
      $
    } (r4)
    (r4) edge[] node[] {
      $
      \begin{pmatrix}
        \mat{I} & \mat{0} \\
        \mat{0} & \mat{A}
      \end{pmatrix}
      $
    } (r5)
    ;
  \end{tikzpicture}
  \caption{\emph{Top}: example of a counter machine
    $\mathcal{M}$. \emph{Bottom}: an affine VASS $\mathcal{V}$
    simulating $\mathcal{M}$.}\label{fig:minsky}
\end{figure}

  Observe that for every $\lambda \in \N$, the following holds:
  \[
  \mat{A} \cdot \lambda \vec{e}_j =
  \begin{cases}
    \vec{0} & \text{if } \lambda = 0, \\
    \lambda \cdot \mat{A}_{\star, j} & \text{otherwise}.
  \end{cases}
  \] Thus, $\mat{A}$ simulates a zero-test as it leaves all counters
  set to zero if counter $j$ holds value zero, and it generates a
  vector with some negative entry otherwise, which is an invalid
  configuration under $\N$-reachability. Figure~\ref{fig:minsky} shows
  how each transition of $\mathcal{M}$ is replaced in $\V$. We are
  done since
  \begin{align*}
    (m, n) \trans{*}_\N (m', n') \text{ in } \mathcal{M}
    \iff (m \cdot \vec{e}_j, n \cdot \vec{e}_j) \trans{*}_\N (m'
    \cdot \vec{e}_j, n' \cdot \vec{e}_j) \text{ in } \V. \tag*{\qedhere}
  \end{align*}
\end{proof}

\begin{prop}\label{prop:undec:two}
  $\reach{\C}$ is undecidable if $\C$ does not only contain transfer
  matrices and does not only contain copy matrices.
\end{prop}

\begin{proof}
  If $\C$ contains a matrix with some negative entry, then we are done
  by Proposition~\ref{prop:undec:neg}. Thus, assume $\C$ only contains
  nonnegative matrices. By Proposition~\ref{prop:row:col:two}, we have
  $\norm{\C} \geq 2$. Let $\mat{C}$ be the matrix obtained for $\C$
  from Lemma~\ref{lem:doubling}. Since $\mat{C} \geq \mat{0}$, we have
  $\mat{C} \cdot \vec{v} \geq \vec{0}$ for every $\vec{v} \geq
  \vec{0}$. Hence, multiplication by $\mat{C}$ is always allowed under
  $\N$-reachability. Thus, the reduction from the Post correspondence
  problem given in Theorem~\ref{thm:undec:normtwo} holds here under
  $\N$-reachability, as the only possibly (relevant) negative values
  arose from $\mat{C}$.
\end{proof}

We may finally prove the last part of our dichotomy:

\begin{thm}\label{thm:undec:nonperm}
  $\reach{\C}$ is undecidable for every class $\C$ with some
  nonpermutation matrix.
\end{thm}

\begin{proof}
  Let $\mat{A} \in \C$ be a matrix which is not a permutation
  matrix. By Propositions~\ref{prop:undec:neg}
  and~\ref{prop:undec:two}, we may assume that $\mat{A}$ is either a
  transfer or a copy matrix. Hence, $\mat{A}$ must have a column or a
  row equal to $\vec{0}$, as otherwise it would be a permutation
  matrix. Thus, we either have $\mat{A}_{\star, i} = \vec{0}$ or
  $\mat{A}_{i, \star} = \vec{0}$ for some $i \in [d]$ where
  $d \defeq \size{\mat{A}}$.

  We show that $\C$ implements resets, \ie, the operation
  $f \colon \Z \to \Z$ such that $f(x)~\defeq~0$. This suffices
  to complete the proof since reachability for VASS with resets is
  undecidable~\cite{AK76}.

  Let $X \defeq \{d + 1, d + 2, \ldots, d + n\}$ be the counters for
  which we wish to implement resets. Let $\mat{A}' \defeq
  \ext{\mat{A}}{n}$ and let $\mat{B}_x \defeq \sigma_x(\mat{A}')$
  where $\sigma_x \defeq (x; i)$. Let $x, y \in X$ be such that $x
  \neq y$. \medskip

  \parag{Case $\mat{A}_{\star, i} = \vec{0}$} We have: $\mat{B}_x \cdot
  \vec{e}_x = (\mat{B}_x)_{\star, x} = \mat{A}'_{\star, i} =
  \vec{0}$. Similarly, it can be shown that $\mat{B}_x \cdot \vec{e}_y
  = \vec{e}_y$. Hence, class $\C$ $0$-implements resets. \medskip

  \parag{Case $\mat{A}_{i, \star} = \vec{0}$} The following holds for
  every $\vec{v} \in \Z^{d + n}$:
  \begin{align*}
    (\mat{B}_x \cdot \vec{v})(x)
    &= \sum_{\ell \in [d + n]} \mat{B}_{x, \ell} \cdot \vec{v}(\ell) \\
    &= \sum_{\ell \in [d + n]} \mat{A}'_{i, \sigma_x(\ell)} \cdot
    \vec{v}(\ell) \\    
    &= 0.
  \end{align*}
  Similarly, $(\mat{B}_x \cdot \vec{v})(y) = \vec{v}(y)$. Hence, class
  $\C$ $?$-imple\-ments resets (see Appendix~\ref{sec:appendix}).
\end{proof}

\section{Parameterization by a system rather than a class}
\label{sec:arb:complexity}
In this section, we consider the (integer or standard) reachability
problem parameterized by a fixed affine VASS, rather than a matrix
class. We show that in contrast to the case of classes, this
parameterization yields an \emph{arbitrary} complexity (up to a
polynomial). More formally, this section is devoted to establishing
the following theorem:

\begin{thm}\label{thm:arb:complexity}
  Let $\emptyset \subset L \subset \{0, 1\}^*$ be a decidable decision
  problem and let $\D \in \{\Z, \N\}$. There exists an affine VASS
  $\V$ such that $L$ and the following problem are interreducible
  under polynomial-time many-one reductions:\medskip

  \begin{tabular}{lp{6cm}}
    \multicolumn{2}{l}{\underline{$\dreach{\V}$}} \\[5pt]
  
    \textsc{Input}: & two configurations $p(\vec{u}), q(\vec{v})$; \\

    \textsc{Decide}: & $p(\vec{u}) \trans{*}_\D q(\vec{v})$ in $\V$? \\
  \end{tabular}
\end{thm}

\noindent In order to show Theorem~\ref{thm:arb:complexity}, we will
prove the following technical lemma:

\begin{restatable}{lemm}{lemPolyReduction}\label{lem:poly:reduction}
  Let $\D \in \{\Z, \N\}$ and let $\varphi \colon \D \to \{0, 1\}$ be
  a nontrivial computable predicate. There exists an affine VASS $\V =
  (d, Q, T)$ and control-states $p, q \in Q$ such that:
  \begin{enumerate}[(a)]
  \item For every vector $\vec{u}$, it is the case that $p(x, \vec{u})
    \trans{*}_\D q(0, \vec{0})$ in $\V$ iff $\varphi(x)$
    holds;\label{itm:a}

  \item Any query ``$r(\vec{v}) \trans{*}_\D r'(\vec{v}')$'' can be
    answered in polynomial time if $r \neq p$ or $r'(\vec{v}') \neq
    q(\vec{0})$.\label{itm:b}
  \end{enumerate}
\end{restatable}

Before proving Lemma~\ref{lem:poly:reduction}, let us see why it
implies the validity of Theorem~\ref{thm:arb:complexity}:

\begin{proof}[Proof of Theorem~\ref{thm:arb:complexity}.]
  Let $f \colon \{0, 1\}^* \to \D$ be a polynomial-time computable and
  invertible bijection.\footnote{If $\D = \N$, we can set $f(w) \defeq
    \mathrm{val}(\texttt{1}w) - 1$, \ie, $f = \{\varepsilon \mapsto 0,
    \texttt{0} \mapsto 1, \texttt{1} \mapsto 2, \texttt{00} \mapsto 3,
    \ldots\}$. If $\D = \Z$, we can set $f(\varepsilon) \defeq 0$ and
    $f(w_1 \cdots w_n) \defeq -1^{w_n} \cdot \mathrm{val}(\texttt{1}
    w_1 \cdots w_{n-1})$, \ie, $f = \{\varepsilon \to 0, \texttt{0}
    \mapsto 1, \texttt{1} \mapsto -1, \texttt{00} \mapsto 2,
    \texttt{01} \mapsto -2, \ldots\}$.} Let $\varphi \colon \D \to
  \{0, 1\}$ be the characteristic function of $L$ encoded by $f$,
  \ie, $w \in L$ iff $\varphi(f(w))$ holds. Let $\V$ be the affine
  VASS given by Lemma~\ref{lem:poly:reduction} for $\varphi$, and let $p,
  q$ be its associated control-states. To reduce $L$ to $\reach{\V}$,
  we construct, on input $w$, the pair $I_w \defeq \langle p(f(w),
  \vec{0}), q(0, \vec{0}) \rangle$. We have $w \in L$ iff
  $\varphi(f(w))$ holds iff $I_w \in \reach{\V}$ by
  Lemma~\ref{lem:poly:reduction}.

  The reduction from $\dreach{\V}$ to $L$ is as follows. On input $I =
  \langle r(\vec{v}), r'(\vec{v}') \rangle$:
  \begin{itemize}
  \item If $r \neq p$ or $r'(\vec{v}') \neq q(\vec{0})$, then we check
    whether $r(\vec{v}) \trans{*}_\D r'(\vec{v}')$ in polynomial time
    and return a positive (resp.\ negative) instance of $L$ if it
    holds (resp.\ does not hold). This is possible by
    Lemma~\ref{lem:poly:reduction} and by nontriviality of $L$.
     
  \item Otherwise, $I = \langle p(x, \vec{u}), q(0, \vec{0}) \rangle$
    for some value $x$ and some vector $\vec{u}$, so it suffices to
    return $f^{-1}(x)$. Indeed, by Lemma~\ref{lem:poly:reduction}, $I
    \in \dreach{\V}$ iff $\varphi(x)$ holds iff $f^{-1}(x) \in
    L$.\qedhere
  \end{itemize}
\end{proof}

In order to prove Lemma~\ref{lem:poly:reduction}, we first establish
the following proposition. Informally, it states that although an
affine VASS cannot evaluate a polynomial $P$ with the mere power of
affine functions, it can evaluate $P$ ``weakly''. Moreover, its
structure is simple enough to answer any reachability query in
polynomial time.

\begin{prop}\label{prop:poly:eval}
  Let $\D \in \{\Z, \N\}$, and let $P$ be a polynomial with integer
  coefficients and $k$ variables over $\D$. There exist an affine VASS
  $\V = (d, Q, T)$ and control-states $p, q \in Q$ such that
  $P(\vec{x}) = 0$ iff $p(\vec{x}, \vec{0}) \trans{*}_\D q(\vec{x},
  \vec{0})$ in $\V$. Moreover, the following properties
  hold:
  \begin{enumerate}[(a)]
  \item No transition of $\V$ alters its $k$ first counters
    (corresponding to $\vec{x}$);
    
  \item Any reachability query ``$r(\vec{v}) \trans{*}_\D r'(\vec{v}')$''
    can be answered in polynomial time.\label{itm:poly:poly}
  \end{enumerate}
\end{prop}

\begin{proof}
  We adapt and simplify a contruction of the authors which was given
  in~\cite{BHMR19} for other purposes. Let us first consider the case
  of a single monomial $P(y_1, \ldots, y_k) = c y_1^{d_1} \cdots
  y_k^{d_k}$ with $c \geq 1$. We claim that the affine VASS $\V$
  depicted in Figure~\ref{fig:poly:eval} satisfies the claim. Its
  counters are $C \defeq Y \cup Y' \cup A$, where
  \begin{align*}
  Y &\defeq \{y_1, \ldots, y_k\} \text{ corresponds to the input
    variables;} \\
  Y' &\defeq \{y_{i, j} : i \in [k], j \in [d_i]\} \text{ are
    auxiliary counters that allow to copy $Y$; and} \\
  A &\defeq \{a_0\} \cup \{a_{i, j} : i \in [k], j \in [d_i]\} \cup
  \{a_\text{out}\} \text{ serve as accumulators.}
  \end{align*}
  
  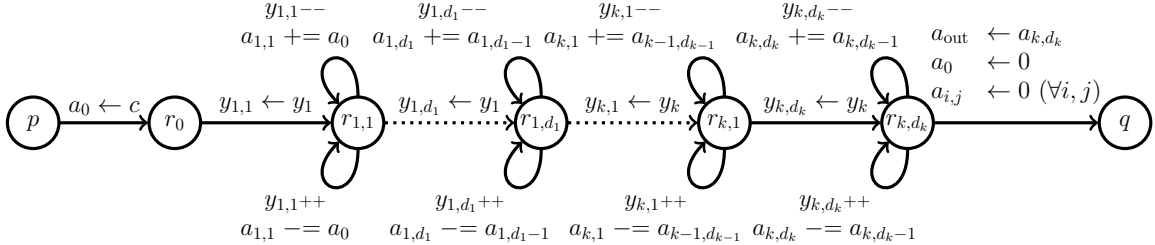
\begin{figure}[h!]
  \centering%
  \newcommand{\inc}[1]{$\text{#1}\scriptstyle\mathrm{++}$}%
  \newcommand{\dec}[1]{$\text{#1}\scriptstyle\mathrm{--}$}%
  \newcommand{\pluseq}{\mathrel{+}=}%
  \newcommand{\minuseq}{\mathrel{-}=}%
  \begin{tikzpicture}[->, auto, very thick, scale=0.65, transform shape, font=\Large]      
    \tikzset{every state/.style={inner sep=1pt, minimum size=30pt}}

    \node[state]                       (p)    {$p$};
    \node[state, right=1.80cm of p]    (r0)   {$r_0$};
    \node[state, right=2.65cm of r0]   (r11)  {$r_{1,1}$};
    \node[state, right=2.65cm of r11]  (r1d1) {$r_{1,d_1}$};
    \node[state, right=2.65cm of r1d1] (rk1)  {$r_{k,1}$};
    \node[state, right=2.65cm of rk1]  (rkdk) {$r_{k,d_k}$};
    \node[state, right=3.35cm of rkdk] (q)    {$q$};

    \path[->]
    (p)    edge node {$a_0 \leftarrow c$} (r0)
    (r0)   edge node {$y_{1,1}   \leftarrow y_1$} (r11)
    (rk1)  edge node {$y_{k,d_k} \leftarrow y_k$} (rkdk)

    (rkdk) edge node[yshift=5pt] {
      \begin{tabular}{ll}
        $a_\text{out}$ & $\leftarrow a_{k, d_k}$ \\
        $a_0$          & $\leftarrow 0$ \\
        $a_{i,j}$      & $\leftarrow 0$ ($\forall i, j$)
      \end{tabular}      
    } (q)

    (r11) edge[out=90, in=135, looseness=8] node[swap, xshift=20pt] {
      \begin{tabular}{c}
        \dec{$y_{1,1}$} \\
        $a_{1,1} \pluseq a_0$
      \end{tabular}
    } (r11)
    (r11) edge[out=-90, in=-135, looseness=8] node[xshift=20pt, yshift=5pt] {
      \begin{tabular}{c}
        \inc{$y_{1,1}$} \\
        $a_{1,1} \minuseq a_0$
      \end{tabular}
    } (r11)
    (r1d1) edge[out=90, in=135, looseness=8] node[swap, xshift=20pt] {
      \begin{tabular}{c}
        \dec{$y_{1,d_1}$} \\
        $a_{1,d_1} \pluseq a_{1,d_1-1}$
      \end{tabular}
    } (r1d1)
    (r1d1) edge[out=-90, in=-135, looseness=8] node[xshift=30pt, yshift=5pt] {
      \begin{tabular}{c}
        \inc{$y_{1,d_1}$} \\
        $a_{1,d_1} \minuseq a_{1,d_1-1}$
      \end{tabular}
    } (r1d1)
    
    (rk1) edge[out=90, in=135, looseness=8] node[swap, xshift=20pt] {
      \begin{tabular}{c}
        \dec{$y_{k,1}$} \\
        $a_{k,1} \pluseq a_{k-1,d_{k-1}}$
      \end{tabular}
    } (rk1)
    (rk1) edge[out=-90, in=-135, looseness=8] node[xshift=35pt, yshift=5pt] {
      \begin{tabular}{c}
        \inc{$y_{k,1}$} \\
        $a_{k,1} \minuseq a_{k-1,d_{k-1}}$
      \end{tabular}
    } (rk1)
    (rkdk) edge[out=90, in=135, looseness=8] node[swap, xshift=20pt] {
      \begin{tabular}{c}
        \dec{$y_{k,d_k}$} \\
        $a_{k,d_k} \pluseq a_{k,d_k-1}$
      \end{tabular}
    } (rkdk)
    (rkdk) edge[out=-90, in=-135, looseness=8] node[xshift=30pt, yshift=5pt] {
      \begin{tabular}{c}
        \inc{$y_{k,d_k}$} \\
        $a_{k,d_k} \minuseq a_{k,d_k-1}$
      \end{tabular}
    } (rkdk)
    ;

    \path[->, dotted]
    (r11)  edge node {$y_{1,d_1} \leftarrow y_1$} (r1d1)
    (r1d1) edge node {$y_{k,1}   \leftarrow y_k$} (rk1)
    ;
  \end{tikzpicture}
  \caption{Affine VASS evaluating $c y_1^{d_1} \cdots
    y_k^{d_k}$, where ``\inc{$x$}'', ``\dec{$x$}'', ``$x
    \pluseq x'$'' and ``$x \minuseq x'$'' stand respectively for
    ``$x \leftarrow x + 1$'', ``$x \leftarrow
    x - 1$'', ``$x \leftarrow x + x'$'' and ``$x \leftarrow x -
    x'$''.}%
  \label{fig:poly:eval}
\end{figure}

  By construction, counters from $Y$ are never altered, merely copied
  onto $Y'$. Moreover, if counters from $Y' \cup A$ initially hold
  zero, and if loops are executed so that counters from $Y'$ reach
  zero, then $a_\text{out}$ contains $P(y_1, \ldots, y_k)$ when
  reaching control-state $q$, and every other auxiliary counter holds
  zero (due to the final resets).

  It remains to argue that Item~\ref{itm:poly:poly} holds. Let us
  consider a query ``$r(\vec{v}) \trans{*}_\D r'(\vec{v}')$''. Observe
  that, although $\V$ has nondeterminism in its loops and uses
  nonreversible operations (copies and resets), it is still
  ``reversible'' since the inputs $Y$ are never altered and since each
  counter from $Y'$ can only be altered at a single copy or via the
  two loops next to it. This provides enough information to answer the
  reachability query. Indeed, we can:
  \begin{itemize}
  \item Answer ``false'' if $\vec{v}$ and $\vec{v}'$ disagree on $Y$;
    
  \item Pretend the accumulators do not exist and traverse $\V$
    backward from $r'(\vec{v}')$ to $r$ by undoing the loops (either
    up only or down only), ensuring each counter from $Y'$ reaches its
    correct value; and answer ``false'' if not possible;

  \item Traverse $\V$ forward from $r(\vec{v})$ by running through the
    loops the number of times identified by the previous traversal;
    this now allows to determine the value of the accumulators $A$;

  \item If the values of $A$ are incorrect in $r'(\vec{v}')$, or if we
    ever drop below zero in the case of $\D = \N$, then we answer
    ``false'', and otherwise ``true''.
  \end{itemize}
  Observe that there is no need to execute the loops one step at a
  time, \eg, if $\vec{v}(y_i) = 10 \geq 3 = \vec{v}'(y_{i, j})$, then
  we can compute $7$ directly rather than in seven steps.

  Let us now consider the general case. We can write the polynomial
  $P$ as \[P(y_1, \ldots, y_k) = \sum_{1 \leq i \leq m} Q_i(y_1,
  \ldots, y_k) - \sum_{m < i \leq n} Q_i(y_1, \ldots, y_k),\] where
  each $Q_i$ is a monomial with positive coefficient. Thus, we can
  compose the above construction sequentially with $n$ distinct sets
  of auxiliary counters (only $Y$ is shared).

  Let $a_{\text{out},i}$ be output counter for $Q_i$. We add a last
  transition that computes \[\sum_{1 \leq i \leq m} a_{\text{out},i} -
  \sum_{m < i \leq n} a_{\text{out},i}\] into an extra counter and
  resets $a_{\text{out},1}, \ldots, a_{\text{out},n}$. This preserves
  all of the above properties. Indeed, each copy variable is still
  only affected locally by a single copy and the next two loops. Note that
  ordering the monomials only matters for $\D = \N$ as summing, \eg,
  $-2 + 5$ blocks, while $5 - 2$ does not.
\end{proof}

We may now conclude by proving Lemma~\ref{lem:poly:reduction}, which
we recall:

\lemPolyReduction*

\begin{proof}
  Since $\varphi$ is decidable, by Matiyasevich's
theorem\footnote{There are two variants of Matiyasevich's theorem:
  stated either over $\N$ or $\Z$. This follows, \eg, from Lagrange's
  four-square theorem, \ie, any number from $\N$ can be written as
  $a^2 + b^2 + c^2 + d^2$ where $a, b, c, d \in \Z$.}~\cite{Mat71}, there exists a
  polynomial $P$, with integer coefficients and $k$ variables, such
  that for every $x \in \D$: \[ \varphi(x) \text{
  holds} \iff \exists \vec{y} \in \D^{k-1} : P(x, \vec{y}) = 0.  \]

  Let $\V'$ be the affine VASS obtained from
  Proposition~\ref{prop:poly:eval} for $P$, and let $p'$ and $q'$ be
  its associated control-states. Let us show that the affine VASS $\V$
  depicted in Figure~\ref{fig:polynomial} satisfies the claim. It uses
  counters $C \defeq Y \cup \overline{Y}$ where $Y \defeq \{y_1,
  \ldots, y_k\}$ are the $k$ first counters of $\V'$ corresponding to
  the variables of $P$; and where $\overline{Y}$ forms the other
  auxiliary counters of $\V'$. We let $m \defeq |C|$ and sometimes
  refer to the counters $C$ as $\{c_1, \ldots, c_m\}$. For the rest of
  the proof, let us write $\vec{u}_X$ to denote the vector obtained by
  restricting $\vec{u}$ to counters $X \subseteq C$.

  The affine VASS $\V$ is divided into two parts connecting $p$ to
  $q$: the \emph{upper gadget} made of control-states $\{a,
  b_1, \ldots, b_m\}$, and the \emph{lower gadget} made of $\V'$.

  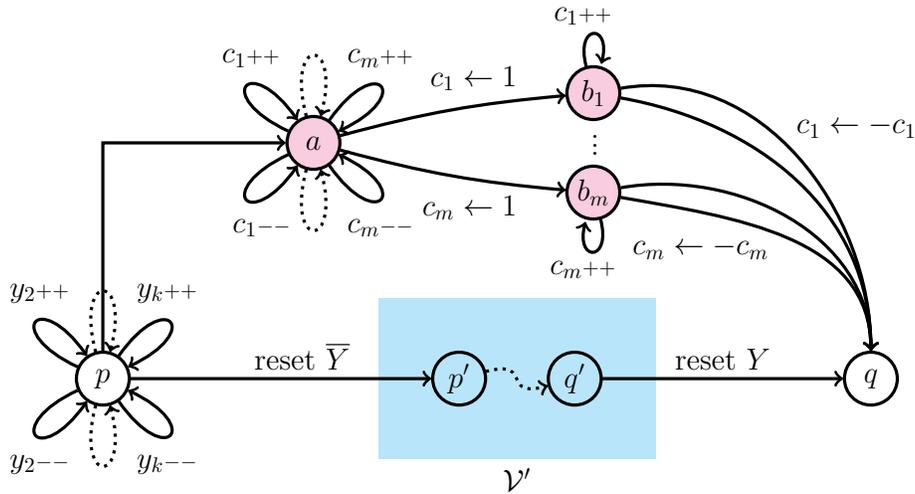
\begin{figure}[h]
  \newcommand{\inc}[1]{$\text{#1}\scriptstyle\mathrm{++}$}%
  \newcommand{\dec}[1]{$\text{#1}\scriptstyle\mathrm{--}$}%
  \centering%
  \vspace*{-5pt} %% Remove extra white space from TikZ
  \begin{tikzpicture}[->, node distance=1cm, auto, very thick, scale=0.8, transform shape, font=\Large]      
    \tikzset{every state/.style={inner sep=1pt, minimum size=25pt}}

    %% Lower gadget
    %% > States
    \colorlet{colLower}{cyan!25}
    
    \node[state] (p)                                  {$p$};
    \node[state, fill=colLower] (pp) [right=5cm of p]  {$p'$};
    \node[state, fill=colLower] (qq) [right=1cm of pp] {$q'$};    
    \node[state]                (q)  [right=4cm of qq] {$q$};

    \begin{scope}[on background layer]
      \node[fit=(pp)(qq), inner sep=20pt, fill=colLower,
           label=below:$\V'$] (lower) {};
    \end{scope}

    %% > Transitions
    %% >> Resets
    \path[->]
    (p) edge node[xshift=10pt] {reset $\overline{Y}$} (pp)
    ;
    
    %% >> Guess y_i
    \path[->, looseness=15]
    (p) edge[out=155, in=125] node[xshift=20pt, yshift=5pt]
         {\inc{$y_2$}} (p)
    (p) edge[out= 55, in= 25] node[xshift=-20pt, yshift=5pt]
         {\inc{$y_k$}} (p)
    (p) edge[out=-155, in=-125] node[swap, xshift= 20pt, yshift=-5pt]
         {\dec{$y_2$}} (p)
    (p) edge[out= -55, in= -25] node[swap, xshift=-20pt, yshift=-5pt]
         {\dec{$y_k$}} (p)
    ;
         
    \path[->, looseness=15, dotted]
    (p) edge[out= 105, in= 75] node[] {} (p)
    (p) edge[out=-105, in=-75] node[] {} (p)
    ;

    %% >> Polynomial weak evaluation
    \path[->]
    (qq) edge node {reset $Y$} (q)
    ;

    \path[->, dotted]
    (pp) edge[out=15, in=195, looseness=2] node {} (qq)
    ;

    %% Upper gadget
    %% > States
    \colorlet{colUpper}{magenta!25}

    \node[state, fill=colUpper]  (a) [above=3cm of p, xshift=3.5cm] {$a$};
    \node[state, fill=colUpper] (b1) [above right=5pt and 4cm of a] {$b_1$};
    \node[yshift=2pt]           (bi) [      right=        4cm of a] {$\vdots$};
    \node[state, fill=colUpper] (bm) [below right=5pt and 4cm of a] {$b_m$};

    %% > Transitions
    \draw[->] (p) |- (a);

    \path[->]
    (a)  edge[bend  left=5] node[xshift=35pt, yshift=+5pt]
         {$c_1 \leftarrow 1$} (b1)
    (a)  edge[bend right=5] node[xshift=35pt, yshift=-5pt, swap]
         {$c_m \leftarrow 1$} (bm)
    ;

    \path[->]
    (b1) edge[out=-10, in=90] node {} (q)
    (b1) edge[out=+10, in=90] node[xshift=-5pt, yshift= 5pt]
         {$c_1 \leftarrow -c_1$} (q)

    (bm) edge[out=+10, in=90] node {} (q)
    (bm) edge[out=-10, in=90] node[xshift=-5pt, yshift=5pt, swap]
         {$c_m \leftarrow -c_m$} (q)

    (b1) edge[loop above] node[xshift=-5pt] {\inc{$c_1$}} ()
    (bm) edge[loop below] node[xshift=-5pt] {\inc{$c_m$}} ()
    ;
    
    \path[->, looseness=15]
    (a) edge[out=155,  in= 125] node[xshift=+20pt, yshift=5pt]
         {\inc{$c_1$}} (a)    
    (a) edge[out= 55,  in=  25] node[xshift=-20pt, yshift=5pt]
         {\inc{$c_m$}} (a)         
    (a) edge[out=-155, in=-125] node[xshift=+25pt, yshift=-5pt, swap]
         {\dec{$c_1$}} (a)
    (a) edge[out=-55,  in= -25] node[xshift=-20pt, yshift=-5pt, swap]
         {\dec{$c_m$}} (a)
    ;

    \path[->, looseness=15, dotted]
    (a) edge[out= 105, in= 75] node[] {} (a)
    (a) edge[out=-105, in=-75] node[] {} (a)
    ;

    %%% Lower to upper
    %\path[->]
    %(P) edge[out=90, in=-98, looseness=1.5] (a)
    %;
  \end{tikzpicture}
  \caption{Affine VASS $\V$ of Lemma~\ref{lem:poly:reduction}, where
    ``\inc{$x$}'', ``\dec{$x$}'' and ``reset $\{x_1, x_2, \ldots\}$''
    stand respectively for ``$x \leftarrow x + 1$'', ``$x \leftarrow x
    - 1$'' and ``$x_1 \leftarrow 0; x_2 \leftarrow 0; \cdots$''.}%
  \label{fig:polynomial}
\end{figure}

  The purpose of the lower gadget is to satisfy Item~\ref{itm:a} by
  checking whether $\varphi(y_1)$ holds. This is achieved by
  (1)~guessing an arbitrary assignment to counters $Y \setminus
  \{y_1\}$; (2)~resetting auxiliary counters $\overline{Y}$;
  (3)~evaluating $P(y_1, \ldots, y_k)$; and (4)~resetting counters
  $Y$.

  The purpose of the upper gadget is to satisfy Item~\ref{itm:b} by
  simplifying other queries from $p$ to $q$. More precisely, the upper
  gadget allows to generate an arbitrary nonzero vector. This is
  achieved by (1)~setting each counter $c_i \in C$ to an arbitrary
  value; (2)~nondeterministically setting a counter $c_j$ to $1$;
  (3)~setting $c_j$ to an arbitrary positive value;
  (4)~nondeterministically keeping or flipping the sign of $c_j$ (the
  latter only works if $\D = \Z$). These steps ensure that all counter
  values are possible, provided that some counter $c_j \neq 0$.

  Let us first show Item~\ref{itm:a}. Suppose $p(x, \vec{u})
  \trans{*}_\D q(0, \vec{0})$. Clearly, the target is not reached
  through the upper gadget. So, the following holds for some value
  $x'$ and some vectors $\vec{y}, \vec{y}'$:
  \[
  p(x, \vec{u}_{Y \setminus \{y_1\}}, \vec{u}_{\overline{Y}})
  \trans{*}_\D  p(x, \vec{y}, \vec{u}_{\overline{Y}})
  \trans{}_\D  p'(x, \vec{y}, \vec{0})
  \trans{*}_\D q'(x', \vec{y}', \vec{0})
  \trans{}_\D   q(0, \vec{0}, \vec{0}).
  \]
  Since $\V'$ does not alter counters $Y$, we have $(x, \vec{y}) =
  (x', \vec{y}')$ and consequently:
  \[
  p'(x, \vec{y}, \vec{0}) \trans{*}_\D
  q'(x, \vec{y}, \vec{0}).
  \]
  Hence, $P(x, \vec{y}) = 0$, which implies that $\varphi(x)$ holds,
  as desired. Conversely, if $\varphi(x)$ holds, then $P(x, \vec{y}) =
  0$ holds for some vector $\vec{y}$. Clearly, we can achieve the
  following:
  \[
  p(x, \vec{u}_{Y \setminus \{y_1\}}, \vec{u}_{\overline{Y}})
  \trans{*}_\D p(x, \vec{y}, \vec{u}_{\overline{Y}})
  \trans{}_\D  \underbrace{p'(x, \vec{y}, \vec{0})
  \trans{*}_\D q'(x, \vec{y}, \vec{0})}_\text{by Prop.~\ref{prop:poly:eval}}
  \trans{}_\D  q(0, \vec{0}, \vec{0}).\]

  Let us now show Item~\ref{itm:b}. Recall that we want to show that
  queries not covered by Item~\ref{itm:a} can be answered in
  polynomial time. These are queries of the form
  $r(\vec{v}) \trans{*}_\D r'(\vec{v}')$ where $\neg(r = p \land
  r'(\vec{v}') = q(\vec{0}))$, which amounts to either $r \neq p$ or
  $r'(\vec{v}') \neq q(\vec{0})$.

  We assume w.l.o.g.\ that $r$ can reach $r'$ in the underlying graph,
  as it can be tested in linear time. Let $Q'$ be the control-states
  of $\V'$. We make a case distinction on $r$ and $r'$, and explain
  each time how to answer the query.\medskip

  \noindent\emph{Lower part}:\medskip
  
  \begin{itemize}
  \setlength\itemsep{5pt}
    
  \item \emph{Case $r = r' = p$.} Amounts to $\vec{v}'(c) =
    \vec{v}(c)$ for every $c \in C \setminus \{y_2, \ldots y_k\}$.

  \item \emph{Case $r = p$, $r' \in Q'$.} Recall that $\V'$ does
    not alter $Y$, that $p$ can generate any values within $Y \setminus
    \{y_1\}$, and that the transition from $p$ to $p'$ resets
    $\overline{Y}$. Hence, we have:
    \[
    p(\vec{v}) \trans{*}_\D r'(\vec{v}')
    \iff
    \vec{v}(y_1) = \vec{v}'(y_1)
    \land
    p'(\vec{v}'_Y, \vec{0}) \trans{*}_\D r'(\vec{v}'_Y,
    \vec{v}'_{\overline{Y}}).
    \]
    The latter can be answered in polynomial time by
    Proposition~\ref{prop:poly:eval}.
    
  \item \emph{Case $r = p$, $r' = q$.} Since $r = p$, we must have
    $\vec{v}' \neq \vec{0}$ by assumption. Thus, the answer is
    ``true'' since the upper gadget allows to reach any nonzero vector
    in $q$.

  \item \emph{Case $r, r' \in Q'$.} Can be tested in polynomial time
    by Proposition~\ref{prop:poly:eval}.
    
  \item \emph{Case $r \in Q'$, $r' = q$.} Recall that $\V'$ does
    not alter $Y$, and that the transition from $q'$ to $q$ resets
    $Y$, but not $\overline{Y}$. Hence, we have:
    \[
    r(\vec{v}) \trans{*}_\D q(\vec{v}')
    \iff
    r(\vec{v}_Y, \vec{v}_{\overline{Y}}) \trans{*}_\D q'(\vec{v}_Y,
    \vec{v}'_{\overline{Y}}).
    \]
    The latter can be answered in polynomial time by
    Proposition~\ref{prop:poly:eval}.

  \end{itemize}

  \medskip\noindent\emph{Upper part:}\medskip
  \begin{itemize}
    \setlength\itemsep{5pt}

  \item \emph{Case $r \in \{p, a\}$ and $r' = a$.} The
    answer is always ``true''.

  \item \emph{Case $r \in \{p, a\}$ and $r' = b_i$.}  Amounts to
    $\vec{v}'(c_i) \geq 1$.

  \item \emph{Case $r = a$, $r' = q$.} Amounts to $\bigvee_{i \in
    [m]} \vec{v}'(c_i) \neq 0$.

  \item \emph{Case $r = r' = b_i$.} Amounts to $\vec{v}'(c_i) \geq
    \vec{v}(c_i)$ and $\vec{v}'(c_j) = \vec{v}(c_j)$ for all $j \neq
    i$.

  \item \emph{Case $r = b_i$, $r' = q$.} Amounts to
    $|\vec{v}'(c_i)| \geq |\max(\vec{v}(c_i), 0)|$ and $\vec{v}'(c_j)
    = \vec{v}(c_j)$ for $j \neq i$. \qedhere
  \end{itemize}
\end{proof}

\section{Conclusion and further work}
\label{sec:conclusion}
Motivated by the use of relaxations to alleviate the tremendous
complexity of reachability in VASS, we have studied the complexity of
integer reachability in affine VASS\@.

Namely, we have shown a trichotomy on the complexity of integer
reachability for affine VASS\@: it is \NP-complete for any class of
reset matrices; \PSPACE-complete for any class of pseudo-transfers
matrices and any class of pseudo-copies matrices; and undecidable for
\emph{any} other class. Moreover, the notions and techniques
introduced along the way allowed us to give a complexity dichotomy for
(standard) reachability in affine VASS\@: it is decidable for any class
of permutation matrices, and undecidable for any other class. This
provides a complete general landscape of the complexity of
reachability in affine VASS\@.

We further complemented these trichotomy and dichotomy by showing
that, in contrast to the case of classes, the (integer or standard)
reachability problem has an arbitrary complexity when it is
parameterized by a fixed affine VASS\@. A further direction of study is
the range of possible complexities for integer reachability relations
for specific matrix monoids, which is entirely open.

Another direction lies in the related \emph{coverability problem}
which asks whether $p(\vec{u}) \trans{*}_\D q(\vec{v}')$ for some
$\vec{v}' \geq \vec{v}$ rather than $\vec{v}' = \vec{v}$. As shown in
the setting of~\cite{HH14,CHH18}, this problem is equivalent to the
reachability problem for $\D = \Z$. Indeed, given an affine VASS $\V$,
it is possible to construct an affine VASS $\V'$ such that
\begin{equation}
  p(\vec{u}) \trans{*}_\Z q(\vec{v}) \text{ in } \V
  \iff
  p(\vec{u}, -\vec{u}) \trans{*}_\Z q(\vec{v}, -\vec{v}) \text{ in }
  \V'.\tag{*}\label{eq:cov:reach}
\end{equation}
This can be achieved by replacing each affine transformation
$\mat{A}\vec{x} + \vec{b}$ of $\V$ by
\[
\begin{bmatrix}
  \mat{A} & \mat{0} \\
  \mat{0} & \mat{A}
\end{bmatrix}
\begin{bmatrix}
  \vec{x} \\
  \vec{x}'
\end{bmatrix}
+
\begin{bmatrix}
   \vec{b} \\
  -\vec{b}
\end{bmatrix}.
\]
Furthermore, note that classes are closed under this construction,
which shows that integer coverability and integer reachability are
equivalent w.r.t.\ classes. However,~\eqref{eq:cov:reach} does not
hold for $\D = \N$. In this case, it is well-known, \eg, that the
coverability problem is decidable for VASS with resets or transfers,
while the reachability problem is not. Hence, a precise
characterization of the complexity landscape remains unknown in this
case.

%% Bibliography
\bibliographystyle{alpha}
\bibliography{references}

%% Appendix
\clearpage
\appendix
\section{Appendix}
\label{sec:appendix}
\subsection{Details for the proof of Proposition~\ref{prop:exist:flip}}
\leavevmode\medskip

\parag{Pseudo-transfer matrix} Let us first prove
properties~\ref{itm:flips:src:tgt} and~\ref{itm:flips:other} stated
within the proof of Proposition~\ref{prop:exist:swaps} for the case
where $\mat{A}$ is a pseudo-transfer matrix. The validity
of~\ref{itm:flips:src:tgt} for $\mat{B}_{s, t}$ follows from:
\begin{align}
  (\mat{B}_{s, t} \cdot \vec{e}_s)(k)
  &= (\mat{B}_{s, t})_{k, s} \nonumber \\
  &= \mat{A}'_{\pi_{s, t}(k), b}
  && \text{(since $\pi_{s, t}(s) = b$)} \nonumber \\
  &= -1 \text{ if } \pi_{s, t}(k) = a \text{ else } 0\label{eq:flip:tr} \\
  &= -1 \text{ if } k = t \text{ else } 0
  && \text{(since $\pi_{s, t}(t) = a$)} \nonumber
  \intertext{where~\eqref{eq:flip:tr} follows from $\mat{A}'_{a, b} =
    -1$ and the fact that $\mat{A}'$ is a pseudo-transfer matrix. The
    validity of~\ref{itm:other} $\mat{B}_{s, t}$ follows from:}
  (\mat{B}_{s, t} \cdot \vec{e}_u)(k) \nonumber
  &= (\mat{B}_{s, t})_{k, u} \nonumber \\
  &= \mat{A}'_{\pi_{s, t}(k), u} = 1
  && \text{(since $u \in X' \setminus \{s, t\}$)} \nonumber \\
  &= 1 \text{ if } \pi_{s, t}(k) = u \text{ else } 0
  && \text{(by def.\ of $\mat{A}'$)} \nonumber \\
  &= 1 \text{ if } k = u \text{ else } 0
  && \text{(since $\pi_{s, t}(u) = u$)}. \nonumber
\end{align}
The same proofs apply mutatis mutandis for $\mat{C}_t$. \medskip

\parag{Pseudo-copy matrix} Let us now prove
Proposition~\ref{prop:exist:flip} for the case where $\mat{A}$ is a
pseudo-copy matrix. For this case, we consider $?$-implementation and
hence $V_X = \Z^{d'}$. Similarly to the case of pseudo-transfer
matrices, we claim that for every $\vec{v} \in V_X$ and every $s, t, u
\in X'$ such that $s \neq t$ and $u \not\in \{s, t\}$, the following
holds:
\begin{enumerate}[(1)]
\item $(\mat{B}_{s, t} \cdot \vec{v})(t) = (\mat{C}_t \cdot
  \vec{v})(t) = -\vec{v}(s)$;\label{itm:src:tgt:flip}
  
\item $(\mat{B}_{s, t} \cdot \vec{v})(u) = (\mat{C}_t \cdot
  \vec{v})(u) = \vec{v}(u)$.\label{itm:other:flip}
\end{enumerate}

Indeed, we have:
\begin{align*}
  (\mat{B}_{s, t} \cdot \vec{v})(t)
  &= \sum_{\ell \in [d] \cup X'} (\mat{B}_{s, t})_{t, \ell} \cdot
  \vec{v}(\ell) \\
  &= \sum_{\ell \in [d] \cup X'} \mat{A}'_{a, \pi_{s, t}(\ell)} \cdot
  \vec{v}(\ell)
  && \text{(since $\pi_{s, t}(t) = a$)} \\
  &= \mat{A}'_{a, b} \cdot \vec{v}(s) + \sum_{\substack{\ell \in [d]
      \cup X' \\ \ell \neq t}} \!\!\mat{A}'_{a, \pi_{s, t}(\ell)} \cdot
  \vec{v}(\ell)
  && \text{(since $\pi_{s, t}(s) = b$)} \\
  &= -\vec{v}(s)
  && \text{(by $\mat{A}'_{a, j} \neq 0 \iff j = b$)},
  \intertext{where the last equality follows from $\mat{A}'_{a, b} =
    -1$ and the fact that $\mat{A}'$ is a pseudo-copy matrix. Moreover,
    we have:}
  (\mat{B}_{s, t} \cdot \vec{v})(u)
  &= \sum_{\ell \in [d] \cup X'} (\mat{B}_{s, t})_{u, \ell} \cdot
  \vec{v}(\ell) \\
  &= \mat{A}'_{u, u} \cdot \vec{v}(u) + \sum_{\substack{\ell \in [d]
      \cup X' \\ \ell \neq u}} \!\!\mat{A}'_{a, \pi_{s, t}(\ell)} \cdot
  \vec{v}(\ell)
  && \text{(since $\pi_{s, t}(u) = u$)} \\
  &= \vec{v}(u)
  && \text{(by $\mat{A}'_{u, j} \neq 1 \iff j = u$)}.
\end{align*}
Again, the same proofs apply mutatis mutandis for $\mat{C}_t$.

We now prove the proposition. Let $x \in X$ and $\vec{v} \in V_X$. We
obviously have $\mat{F}_x \cdot \vec{v} \in V_X$. If $a \neq b$, we
have:
\begin{align*}
  (\mat{F}_x \cdot \vec{v})(x)
  &= (\mat{B}_{z, x} \cdot (\mat{B}_{y, z} \cdot \mat{B}_{x, y}\cdot
  \vec{v}))(x_)
  && \text{(by def.\ of $\mat{F}_x$)} \\
  &= -(\mat{B}_{y, z} \cdot (\mat{B}_{x, y} \cdot \vec{v}))(z)
  && \text{(by~\ref{itm:src:tgt:flip})} \\
  &= (\mat{B}_{x, y} \cdot \vec{v})(y) &&
  \text{(by~\ref{itm:src:tgt:flip})} \\
  &= -\vec{v}(x) &&
  \text{(by~\ref{itm:src:tgt:flip})},
  \intertext{and if $a = b$, we have:}
  (\mat{F}_x \cdot \vec{v})(x)
  &= (\mat{C}_x \cdot \vec{v})(x)
  && \text{(by def.\ of $\mat{F}_x$)} \\
  &= -\vec{v}(x)
  && \text{(by~\ref{itm:src:tgt:flip})}.
\end{align*}
Similarly, by applying~\ref{itm:other:flip} repeatedly, we derive
$(\mat{F}_x \cdot \vec{v})(y) = \vec{v}(y)$ for every $y \in X
\setminus \{x\}$.\hfill\qed{}

\subsection{Details for the proof of Proposition~\ref{prop:exist:swaps}}

The details of the proof are similar to those of the proof of
Proposition~\ref{prop:exist:flip}.\medskip

\parag{Transfer matrix} Let us first prove
properties~\ref{itm:src:tgt} and~\ref{itm:other} stated within the
proof of Proposition~\ref{prop:exist:swaps} for the case where
$\mat{A}$ is a transfer matrix. The validity of~\ref{itm:src:tgt}
follows from:
\begin{align*}
  (\mat{B}_{s, t} \cdot \vec{e}_s)(k) = 1
  &\iff (\mat{B}_{s, t})_{k, s} = 1 \\
  &\iff \mat{A}'_{\pi_{s, t}(k), b} = 1
  && \text{(since $\pi_{s, t}(s) = b$)} \\
  &\iff \pi_{s, t}(k) = a
  && \text{(since $\mat{A}$ is a transfer matrix)} \\
  &\iff k = t.
  \intertext{The validity of~\ref{itm:other} follows from:}
  (\mat{B}_{s, t} \cdot \vec{e}_u)(k) = 1
  &\iff (\mat{B}_{s, t})_{k, u} = 1 \\
  &\iff \mat{A}'_{\pi_{s, t}(k), u} = 1
  && \text{(since $u \in X' \setminus \{s, t\}$)} \\
  &\iff \pi_{s, t}(k) = u
  && \text{(by def.\ of $\mat{A}'$)} \\
  &\iff k = u.
\end{align*}

\parag{Copy matrix} Let us now prove
Proposition~\ref{prop:exist:swaps} for the case where $\mat{A}$ is a
copy matrix. For this case, we consider $?$-implementation and hence
$V_X = \Z^{d'}$. Similarly to the case of transfer matrices, we claim
that for every $\vec{v} \in V_X$ and every $s, t, u \in X'$ such that
$s \neq t$ and $u \not\in \{s, t\}$, the following holds:
\begin{enumerate}[(1)]
\item $(\mat{B}_{s, t} \cdot \vec{v})(t) =
  \vec{v}(s)$;\label{itm:src:tgt:copy}
  
\item $(\mat{B}_{s, t} \cdot \vec{v})(u) =
  \vec{v}(u)$.\label{itm:other:copy}
\end{enumerate}

Indeed, we have:
\begin{align*}
  (\mat{B}_{s, t} \cdot \vec{v})(t)
  &= \sum_{\ell \in [d] \cup X'} (\mat{B}_{s, t})_{t, \ell} \cdot
  \vec{v}(\ell) \\
  &= \sum_{\ell \in [d] \cup X'} \mat{A}'_{a, \pi_{s, t}(\ell)} \cdot
  \vec{v}(\ell)
  && \text{(since $\pi_{s, t}(t) = a$)} \\
  &= \mat{A}'_{a, b} \cdot \vec{v}(s) + \sum_{\substack{\ell \in [d]
      \cup X' \\ \ell \neq t}} \!\!\mat{A}'_{a, \pi_{s, t}(\ell)} \cdot
  \vec{v}(\ell)
  && \text{(since $\pi_{s, t}(s) = b$)} \\
  &= \vec{v}(s)
  && \text{(by $\mat{A}'_{a, j} = 1 \iff j = b$)},
  \intertext{where the last equality follows from $\mat{A}'_{a, b} = 1$
    and the fact that $\mat{A}'$ is a copy matrix. Morever, we have:}
  (\mat{B}_{s, t} \cdot \vec{v})(u)
  &= \sum_{\ell \in [d] \cup X'} (\mat{B}_{s, t})_{u, \ell} \cdot
  \vec{v}(\ell) \\
  &= \mat{A}'_{u, u} \cdot \vec{v}(u) + \sum_{\substack{\ell \in [d]
      \cup X' \\ \ell \neq u}} \!\!\mat{A}'_{a, \pi_{s, t}(\ell)} \cdot
  \vec{v}(\ell)
  && \text{(since $\pi_{s, t}(u) = u$)} \\
  &= \vec{v}(u)
  && \text{(by $\mat{A}'_{u, j} = 1 \iff j = u$)}.
\end{align*}

We may now prove the proposition. Let $\vec{v} \in V_X$ and let $x, y
\in X$ be such that $x \neq y$. We obviously have $\mat{F}_{x, y}
\cdot \vec{v} \in V_X$. Moreover, we have:
\begin{align*}
  (\mat{F}_{x, y} \cdot \vec{v})(x)
  &= (\mat{B}_{z, x} \cdot (\mat{B}_{x, y} \cdot \mat{B}_{y,
    z} \cdot \vec{v}))(x)
  && \text{(by def.\ of $\mat{F}_{x, y}$)} \\
  &= (\mat{B}_{x, y} \cdot (\mat{B}_{y, z} \cdot \vec{v}))(z)
  && \text{(by~\ref{itm:src:tgt:copy})} \\
  &= (\mat{B}_{y, z} \cdot \vec{v})(z) &&
  \text{(by~\ref{itm:other:copy})} \\
  &= \vec{v}(y) &&
  \text{(by~\ref{itm:src:tgt:copy})},
\end{align*}
and symmetrically $(\mat{F}_{x, y} \cdot \vec{v})(y) =
\vec{v}(x)$. Similarly, by applying~\ref{itm:other:copy} three
times, we derive $(\mat{F}_{x, y} \cdot \vec{v})(u) = \vec{v}(u)$ for
every $u \in X \setminus \{x, y\}$.\hfill\qed{}

\subsection{Details for the proof of Proposition~\ref{prop:row:col:two}}

Let us prove the technical lemma invoked within the proof of
Proposition~\ref{prop:row:col:two}:

\begin{lem}\label{lem:same:sign}
  For every class $\C$, if $\C$ contains a matrix which has a row
  (resp.\ column) with at least two nonzero elements, then $\C$ also
  contains a matrix which has a row (resp.\ column) with at least two
  nonzero elements with the same sign.
\end{lem}

\begin{proof}
  We first consider the case of rows. Let $\mat{A} \in \C$, $i, j, k
  \in [d]$ and $a, b \neq 0$ be such that $\mat{A}_{i, j} = a$,
  $\mat{A}_{i, k} = b$ and $j \neq k$. If $a$ and $b$ have the same
  sign, then we are done. Thus, assume without loss of generality that
  $a < 0$ and $b > 0$.

  Let us first give an informal overview of the proof where we see
  $\mat{A}$ as an operation over some counters. We have two counters
  $x$ and $y$ that we wish to sum up (with some positive integer
  coefficients), using a supply of counters set to zero. We apply
  $\mat{A}$ to $x$ and some zero counters to produce $a \cdot x$ in
  some counter (while discarding extra noise into some other ones),
  and we then apply $\mat{A}$ again to $a \cdot x$, $y$ and some zero
  counters in such a way that we get $a^2 \cdot x + b \cdot y$. The
  matrix achieving this procedure will have $a^2$ and $b$ on a common
  row.

  More formally, we wish to obtain a matrix $\mat{D}$ with positive
  entries $a^2$ and $b$, and more precisely such that $\mat{D}_{i, j'}
  = a^2$ and $\mat{D}_{i, k} = b$ for some counter $j'$. Let $\mat{B}
  \defeq \ext{\mat{A}}{d}$, $\mat{C} \defeq \perm{\mat{B}}{\sigma}$
  and $\mat{D} \defeq \mat{B} \cdot \mat{C}$ where $\sigma \colon [2d]
  \to [2d]$ is the following permutation:
  \[
  \sigma \defeq
  \begin{cases}
    (j; i; i+d) \cdot \prod_{\ell \in [d] \setminus \{i, j\}}
    (\ell; \ell + d) & \text{if } j \neq i, \\
    \prod_{\ell \in [d] \setminus \{i\}} (\ell; \ell + d) & \text{if
    } j = i.
  \end{cases}
  \]
  We claim that $\mat{D}$ is as desired. First, observe that:
  \begin{align*}
    \mat{D}_{i, k}
    &= \sum_{\ell \in [2d]} \mat{B}_{i, \ell} \cdot
    \mat{B}_{\sigma(\ell), k+d}
    && \text{(by def.\ of $\mat{D}$ and $\sigma(k) = k+d$)} \\
    &= \mat{B}_{i, k} \cdot \mat{B}_{k+d, k+d}\label{eq:ik}
    && \text{(since $\mat{B}_{\sigma(\ell), k+d} \neq 0 \iff
      \sigma(\ell) = k+d$)} \\
    &= b
    && \text{(since $\mat{B}_{k+d, k+d} = 1$)}.
  \end{align*}  
  Thus, $\mat{D}$ has a positive entry on row $i$. It remains to show
  that $\mat{D}$ has another positive entry on row $i$. We make a case
  distinction on whether $j = i$.\medskip

  \parag{Case $j \neq i$} Note that $k \neq i+d$. Hence, we are done
  since:
  \begin{align*}
    \mat{D}_{i, i+d}
    &= \sum_{\ell \in [2d]} \mat{B}_{i, \ell} \cdot \mat{B}_{\sigma(\ell), j}
    && \text{(by def.\ of $\mat{D}$ and $\sigma(i+d) = j$)} \\
    &= \sum_{\ell\, :\, \ell, \sigma(\ell) \in [d]} \mat{B}_{i, \ell} \cdot
    \mat{B}_{\sigma(\ell), j}
    && \text{(since $\mat{B} = \ext{\mat{A}}{d}$ and $i, j \in [d]$)} \\
    &= \mat{B}_{i, j} \cdot \mat{B}_{i, j}
    && \text{(by def.\ of $\sigma$)} \\
    &= a^2.
  \end{align*}
  
  \parag{Case $j = i$} Note that $k \neq j = i$. Hence, we are done
  since:
  \begin{align*}
    \mat{D}_{i, i}
    &= \sum_{\ell \in [2d]} \mat{B}_{i, \ell} \cdot \mat{B}_{\sigma(\ell), i}
    && \text{(by def.\ of $\mat{D}$ and $\sigma(i) = i$)} \\
    &= \sum_{\ell : \ell, \sigma(\ell) \in [d]} \mat{B}_{i, \ell} \cdot
    \mat{B}_{\sigma(\ell), i}
    && \text{(since $\mat{B} = \ext{\mat{A}}{d}$ and $i \in [d]$)} \\
    &= \mat{B}_{i, i} \cdot \mat{B}_{i, i}
    && \text{(by def.\ of $\sigma$)} \\
    &= a^2
    && \text{(since $i = j$)}.
  \end{align*}
  We are done proving the proposition for the case of rows.

  For the case of columns, we can instead assume that
  $\mat{A}^\transpose \in \C$, \ie, the transpose of $\mat{A}$ belongs
  to $\C$. Since $\mat{D}^\transpose$ is as desired, we simply have to
  show that $\mat{D}^\transpose \in \C$. This is the case since:
  \begin{align*}
    \mat{D}^\transpose
    &= (\mat{B} \cdot \mat{C})^\transpose \\
    &= \mat{C}^\transpose \cdot \mat{B}^\transpose \\
    &= (\mat{P}_\sigma \cdot \mat{B} \cdot
    \mat{P}_{\sigma^{-1}})^\transpose \cdot \mat{B}^\transpose
    && \text{(since $\mat{C} = \perm{\mat{B}}{\sigma}$)} \\
    &= (\mat{P}_\sigma \cdot \mat{B}^\transpose \cdot
    \mat{P}_{\sigma^{-1}}) \cdot \mat{B}^\transpose
    && \text{(since $\mat{P}_{\pi^{-1}} = \mat{P}_\pi^\transpose$ for
      every perm.\ $\pi$)} \\    
    &= \perm{\mat{B}^\transpose}{\sigma} \cdot
    \mat{B}^\transpose \\
    &= \perm{(\ext{\mat{A}}{d})^\transpose}{\sigma} \cdot
    (\ext{\mat{A}}{d})^\transpose \\
    &= \sigma(\ext{(\mat{A}^\transpose)}{d}) \cdot
    \ext{(\mat{A}^\transpose)}{d}
    && \text{(since $(\ext{\mat{A}}{d})^\transpose =
      \ext{(\mat{A}^\transpose)}{d}$)}. \\
    &\in \C
    && \text{(since $\mat{A}^\transpose \in \C$)}. \qedhere
  \end{align*}
\end{proof}

\subsection{Details for the proof of Theorem~\ref{thm:undec:nonperm}}

We prove the missing details for both cases:\medskip

\parag{Case $\mat{A}_{\star, i} = \vec{0}$} We have $\mat{B}_x \cdot
\vec{e}_y = \vec{e}_y$ since:

\begin{align*}
  (\mat{B}_x \cdot \vec{e}_y)(k)
  &= (\mat{B}_x){k, y} \\
  &= \mat{A}'_{\sigma_x(k), y}
  && \text{(since $y \neq x$)} \\
  &= 1 \iff k = y
  && \text{(by def.\ of $\mat{A}'$ and $\sigma_x$).}
  \intertext{\parag{Case $\mat{A}_{i, \star} = \vec{0}$} We have:}
  (\mat{B}_x \cdot \vec{v})(y)
  &= \sum_{\ell \in [d + n]} \mat{B}_{y, \ell} \cdot \vec{v}(\ell) \\
  &= \mat{A}'_{y, y} \cdot \vec{v}(y)
  && \text{(by $y \neq x$ and def.\ of $\mat{A}'$ and $\sigma_x$)} \\
  &= \vec{v}(y).
  && \tag*{\qed{}}
\end{align*}

\end{document}